\documentclass[runningheads]{./llncs}
\usepackage{lineno}
\usepackage[misc,geometry]{ifsym}
\usepackage{tikz}
\usetikzlibrary{calc,shapes.multipart,chains,arrows,positioning,decorations.pathreplacing,calligraphy}

\usepackage{cite}
\usepackage{booktabs}
\usepackage[plain]{algorithm}
\usepackage{subcaption}
\DeclareCaptionSubType*{algorithm}

\DeclareCaptionLabelFormat{alglabel}{Alg.~#2}
\usepackage{xspace}
\usepackage{amssymb}
\usepackage{amsmath}
\usepackage[inline]{enumitem}
\usepackage[binary-units]{siunitx}
\DeclareSIUnit\linesofcode{loc}
\usepackage{tikz}
\usepackage{pgfplots, pgfplotstable}
\usepackage{verbatim}
\usepackage{microtype}
\usepackage[draft]{fixme}
\usepackage{wrapfig}
\usepackage{multirow}

\usepackage{hyperref}
\hypersetup{%
  colorlinks=true,%
  citecolor=blue,%
  urlcolor=blue,%
  linkcolor=blue,%
}

\usepackage[normalem]{ulem}

\newcommand*{\github}{\textsf{github}\xspace}

\usepackage[noend]{algpseudocode}
\algnewcommand\algorithmicinput{\textbf{Input:}}
\algnewcommand\INPUT{\item[\algorithmicinput]}
\algnewcommand\algorithmicoutput{\textbf{Output:}}
\algnewcommand\Output{\item[\algorithmicoutput]}
\algrenewcommand\algorithmicindent{1ex}

\def\anno#1{{\ooalign{\hfil\raise.07ex\hbox{\small{\rm
          \textcolor{red}{{\tiny #1}}}}\hfil%
        \crcr{\scriptsize \textcolor{blue}{\mathhexbox20D}}}}}

\usepackage{todonotes}

\usepackage{listings}

\definecolor{ckeyword}{HTML}{7F0055}
\definecolor{ccomment}{HTML}{3F7F5F}
\definecolor{cnumber}{HTML}{2A0099}

\lstdefinelanguage{Solidity} {
  keywords={typeof, modifier, function, public, returns, external,
    unchecked, return, for, 
  contract, new, true, false, private, view, pure, memory, catch, function, null, throw, catch, switch, var, if, in, while, do, else, case, break},
  ndkeywords={bool, struct, address, uint256, uint64, mapping, bytes32, string},
  identifierstyle=\color{black},
  sensitive=false,
  comment=[l]{//},
  morecomment=[s]{/*}{*/},
  commentstyle=\color{ccomment}\ttfamily,
  string=[b]",
  showstringspaces=false,
  morestring=[b]',
  showspaces=false,
  showtabs=false,  breaklines=true,
  morekeywords={function, contract, returns},
  breakatwhitespace=true,
  lineskip=-0.6pt,
  basewidth={0.54em, 0.4em},
  basicstyle=\scriptsize\sffamily,
  keywordstyle={\color{ckeyword}\scriptsize\bfseries},
  ndkeywordstyle={\color{black}\scriptsize\bfseries},
  commentstyle={\color{ccomment}\itshape},
  stringstyle={\color{pgreen}},
  numbersep=3pt,
  numberstyle={\tiny\color{cnumber}\ttfamily},
  moredelim=[il][\textcolor{pgrey}]{$$},
  moredelim=[is][\textcolor{pgrey}]{\%\%}{\%\%},
}

\lstdefinelanguage{evm}{
  stringstyle=\scriptsize\ttfamily
  basicstyle=\scriptsize\ttfamily
  keywordstyle=\scriptsize\ttfamily \mdseries
  morekeywords={STOP, ADD, MUL, SUB, DIV, SDIV, MOD, SMOD, ADDMOD, MULMOD, EXP, SIGNEXTEND,
  LT, GT, SLT, SGT, EQ, ISZERO, AND, OR, XOR, NOT, BYTE, SHL, SHR, SAR, SHA3,
  ADDRESS, BALANCE, ORIGIN, CALLER, CALLVALUE, CALLDATALOAD, CALLDATASIZE,
  CALLDATACOPY, CODESIZE, CODECOPY, GASPRICE, EXTCODESIZE, EXTCODECOPY,
  RETURNDATASIZE, RETURNDATACOPY, EXTCODEHASH,
  BLOCKHASH, COINBASE, TIMESTAMP, NUMBER, DIFFICULTY, GASLIMIT,
  POP, MLOAD, MSTORE, MSTORE8, SLOAD, SSTORE, JUMP, JUMPI, PC, MSIZE, GAS, JUMPDEST,
  PUSH, SWAP, SWAP1,SWAP2,SWAP3,SWAP4,SWAP5,SWAP6,SWAP7,SWAP8,SWAP9,SWAP10,SWAP11,
  SWAP12,SWAP13, SWAP14,SWAP15,SWAP16, DUP, DUP1,DUP2,DUP3,DUP4,DUP5,DUP6,DUP7,
  DUP8,DUP9,DUP10,DUP11,DUP12,DUP13, DUP14,DUP15,DUP16,
  LOG0, LOG1, LOG2, LOG3, LOG4, JUMPTO, JUMPIF, JUMPV, JUMPSUB, JUMPSUBV, BEGINSUB, BEGINDATA,
  RETURNSUB, PUTLOCAL, GETLOCAL, CREATE, CALL, CALLCODE, RETURN, DELEGATECALL, CREATE2,
  STATICCALL, REVERT, INVALID, SELFDESTRUCT, LOG, NOP},
  alsoletter={0,1,2,3,4,5,6,7,8,9,a,b,c,d,e,f},
  basicstyle=\scriptsize\ttfamily,
  comment=[l]{//},
  commentstyle={\itshape},
  numbers=none
}
\newcommand{\sscriptcode}[1]{\lstinline[language=Solidity,basicstyle=\scriptsize\ttfamily]{#1}}
\newcommand{\scriptcode}[1]{\sscriptcode{#1}}
\newcommand{\scode}[1]{\lstinline[language=Solidity,basicstyle=\footnotesize\ttfamily]{#1}}
\newcommand{\code}[1]{\scode{#1}}
\newcommand{\mcode}[1]{\text{\small{\texttt{#1}}}}
\newcommand{\smcode}[1]{\text{\scriptsize{\texttt{#1}}}}
\newcommand{\evmcode}[1]{\text{\texttt{\tiny{#1}}}}
\newcommand{\nodetext}[1]{{\text{\scriptsize{#1}}}}

\newcommand{\lstcode}[1]{{\small\texttt{#1}}}
\newcommand{\baseref}[0]{\ensuremath{\textit{baseref}}\xspace}

\newcommand{\lst}[1]{\mbox{\lstinline!#1!}}

\newcommand{\lub}{\sqcup}

\newcommand{\EVM}{EVM\xspace}
\newcommand{\evm}{EVM\xspace}
\newcommand{\dom}[0]{\ensuremath{\mathit{dom}}}

\newcommand{\mem}[1]{{\small $ mem\langle\text{#1}\rangle$}}
\newcommand{\slot}[1]{\textit{s$_{#1}$}}
\newcommand{\slotcc}[1]{{\textcolor{red}{\slot{#1}}}}
\newcommand{\opened}{\mathcal{S}}
\newcommand{\closed}{\mathcal{S}_{all}}
\newcommand{\closing}{\mathcal{K}}
\newcommand{\getslot}{\textit{get\_slots}}
\newcommand{\allslots}{\mathcal{A}}

\newcommand{\stackset}{\mathcal{T}}
\newcommand{\ppmload}{\mathcal{P}_{R}}
\newcommand{\ppmstore}{\mathcal{P}_{W}}

\newcommand{\writeset}{{\small \mathcal{W}}}
\newcommand{\readset}{{\small \mathcal{R}}}
\newcommand{\writingleaks}{{\small \mathcal{N}}}
\newcommand{\reachable}{reachable}
\newcommand{\existsread}{exists\_read}
\newcommand{\set}[1]{\ensuremath{\{#1\}}}
\newcommand{\memread}{\ensuremath{\textit{mr}}}
\newcommand{\memwrite}{\ensuremath{\textit{mw}}}
\usepackage{mathtools}

\newcommand*{\numAddresses}{4,891\xspace}
\newcommand*{\numTimeouts}{626\xspace}
\newcommand*{\numContracts}{19,199\xspace}
\newcommand*{\totalTime}{12,803s\xspace}

\newcommand*{\needlessAccess}{6,242\xspace}
\newcommand*{\numAbsSlots}{679,517\xspace}

\newcommand{\secbeg}{\vspace*{0cm}}

\def\orcidID#1{\smash{\href{http://orcid.org/#1}{\protect\raisebox{-1.25pt}{\protect\includegraphics{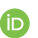}}}}}
\makeatother

\begin{document}

\title{Inferring Needless Write Memory Accesses\\ on Ethereum
  Bytecode (Extended Version)\thanks{This work was funded partially by the Spanish MCI,
    AEI and FEDER (EU) projects PID2021-122830OB-C41 and
    PID2021-122830OA-C44}}


%
\titlerunning{Inferring Needless Write Memory Accesses on Ethereum
  Bytecode (Ext. V.)}
%

%
\author{ Elvira Albert$^{1}$\orcidID{0000-0003-0048-0705}
  \and Jes\'us Correas$^{1}$\orcidID{0000-0002-3219-0799}
   \and Pablo Gordillo$^{1}$(\Letter)\orcidID{0000-0001-6189-4667} \and \\
  Guillermo Rom\'an-D\'iez$^{2}$\orcidID{0000-0002-5427-8855}
  \and Albert Rubio$^{1}$\orcidID{0000-0002-0501-9830}
}
\authorrunning{E. Albert et al.}
\institute{
 Complutense University of Madrid,  Spain \\  
 \and Universidad Polit\'ecnica de Madrid, Spain\\
\email{pabgordi@ucm.es}}

\maketitle              

\secbeg\secbeg\secbeg\secbeg
\begin{abstract}
  Efficiency is a fundamental property of any type of program, but it
  is even more so in the context of the programs executing on the
  blockchain (known as \emph{smart contracts}). This is because
  optimizing smart contracts has direct consequences on reducing the
  costs of deploying and executing the contracts, as there are fees to
  pay related to their bytes-size and to their resource consumption
  (called \emph{gas}). Optimizing memory usage is considered a
  challenging problem that, among other things, requires a precise
  inference of the memory locations being accessed.  This is also the
  case for the Ethereum Virtual Machine (EVM) bytecode generated by
  the most-widely used compiler, \texttt{solc}, whose rather
  unconventional and low-level memory usage challenges automated
  reasoning.  This paper presents a static analysis, developed at the
  level of the EVM bytecode generated by \texttt{solc}, that infers
  write memory accesses that are needless and thus can be safely
  removed. The application of our implementation on more than 19,000
  real smart contracts has detected about 6,200 needless write
  accesses in less than 4 hours. Interestingly, many of these writes
  were involved in memory usage patterns
  generated by \texttt{solc} that can be greatly optimized by removing entire blocks of bytecodes.  To the
  best of our knowledge, existing optimization tools cannot infer such
  needless write accesses, and hence cannot detect these inefficiencies that affect both the deployment and the execution costs of Ethereum smart contracts.
  %


\end{abstract}
%


\secbeg
\secbeg
\secbeg\secbeg
\secbeg
\secbeg
\section{Introduction}\label{sec:introduction}
\secbeg

\paragraph{EVM and memory model.} Ethereum~\cite{yellow} is considered
the world-leading programmable blockchain today. It provides a virtual
machine, named EVM (Ethereum Virtual Machine)~\cite{evmBook}, to
execute the programs that run on the blockchain. Such programs, known
as Ethereum ``smart contracts'', can be written in high-level
programming languages such as Solidity~\cite{solidity-doc},
Vyper~\cite{vyper}, Serpent~\cite{serpent} or Bamboo~\cite{bamboo} and
they are then compiled to EVM bytecode. The EVM bytecode is the code
finally deployed in the blockchain, and has become a uniform format to
develop analysis and optimization tools.  The memory model of EVM
programs has been described in previous work
\cite{yellow,LagouvardosGTS20,HajduJ20,TsankovDDGBV18}. Mainly, there
are three regions in which data can be stored and accessed: (1) The
EVM is a stack-based virtual machine, meaning that most instructions
perform computations using the topmost elements in a machine
\emph{stack}. This memory region can only hold a limited amount of
values, up to 1024 256-bit words.
%
(2) EVM programs store data persistently using a memory region named
\emph{storage} that consists of a mapping of 256-bit addresses to
256-bit words and whose contents persist between external function
calls.
  %
(3) The third memory region is a local volatile memory area that we
will refer to as EVM \emph{memory}, and which is the focus of our
work. This memory area behaves as a simple word-addressed array of
bytes that can be accessed by byte or as a one-word group.  The EVM
memory can be used to allocate dynamic local data (such as arrays or
structs) and also for specific EVM bytecode instructions which have
been designed to require some lengthy operands to be stored in local
memory. This is the case of the instructions for computing
cryptographic hashes, or for passing arguments to and returning data
from external function calls.
Compilers use the stack and volatile memory regions in different
ways. The most-used Solidity compiler \texttt{solc} generates EVM code
that uses the stack for storing value-type local variables, as well as
intermediate values for complex computations and jump addresses,
whereas reference-type local variables such as array types and
user-defined struct types are located in memory.
The focus of our work is on detecting needless write memory accesses
on the code generated by 
\texttt{solc}. Nevertheless, as the analysis works at EVM level, it
could be easily adapted to EVM code generated by any other compiler.

\paragraph{Optimization.}  Optimization of Ethereum smart contracts is
a hot research topic, see
e.g.~\cite{AlbertGHR22,DBLP:conf/icse/ChenLZCLLZ18,ChenFLZLLXCZ20,ebso,2020_brandstatter_et_al,2020_schett_et_al,GaoSSLSB21}
and their references. This is because the reduction of their costs is
relevant for three reasons: (1) \emph{Deployment fees}. When the
contract is deployed on the blockchain, the owner pays a fee related
to the size in bytes of the bytecode. Hence, a clear optimization
criterion is the bytes-size of the program. The Solidity compiler
\texttt{solc}~\cite{solidity-doc} has as optimization target such
bytes-size reduction.  (2) \emph{Gas-metered execution}. There is a
fee to be paid by each client to execute a transaction in the
blockchain. This fee is a fixed amount per transaction plus the cost
of executing all bytecode instructions within the function being
invoked within the transaction. This cost is measured in ``gas''
(which is then priced in the corresponding cryptocurrency) and this is
why the execution is said to be gas-metered.  The EVM specification
(\hspace{1sp}\cite{yellow} and more recent updates) provides a precise
gas consumption for each bytecode instruction in the language. The
goal of most EVM bytecode optimization
tools~\cite{AlbertGHR22,DBLP:conf/icse/ChenLZCLLZ18,ChenFLZLLXCZ20,ebso,2020_brandstatter_et_al,GaoSSLSB21}
is to reduce such gas consumption, as this will revert on reducing the
price of all transactions on the smart contract.  (3) \emph{Enlarging
  Ethereum's capability}. Due to the huge volume of transactions that
are being demanded, there is a huge interest in enlarging the
capability of the Ethereum network to increase the number of
transactions that can be handled. Optimization of EVM bytecode in
general --and of its memory usage in particular-- is an important step
contributing into this direction.

\paragraph{Challenges and contributions.} Optimizing memory usage is
considered a challenging problem that requires a precise inference of
the memory locations being accessed, and that usually varies according
to the memory model of the language being analyzed, and to the compiler
that generates the code to be executed. In the case of Ethereum smart
contracts generated by the \texttt{solc} compiler, %
the memory model is rather unconventional and its low-level memory
usage patterns challenge automated reasoning. On one hand, instead of
having an instruction to allocate memory, the allocation is performed
by a sequence of instructions that use the value stored at address
$0x40$ as the \emph{free memory pointer}, i.e., a pointer to the first memory
address available for allocating new memory. In the general case, the
memory is structured as a sequence of \emph{slots}: a slot is composed
of several consecutive memory locations that are accessed in the
bytecode from the same initial memory location plus a corresponding
offset.  A slot might just hold a data structure created in the smart
contract but also, when nested data structures are used, from one slot
we can find pointers to other memory slots for the nested components.
Finally, there are other type of 
\emph{transient} slots that hold temporary data and that need to be
captured by a precise memory analysis as well. These features pose the
main challenges to infer needless write accesses and, to handle them
accurately, we make the following main contributions:
(1) we present a \emph{slot analysis} to (over-)approximate the slots
created along the execution and the program points at which they are
allocated;
(2) we then introduce a \emph{slot usage analysis} which infers the
accesses to the different slots from the bytecode instructions;
(3) we finally infer \emph{needless write accesses}, i.e., program
points where the memory is written but is never read by any subsequent
instruction of the program; and
(4) we implement the approach and perform a thorough experimental
evaluation on real smart contracts detecting needless write accesses
which belong to highly optimizable memory usage patterns generated by
\texttt{solc}.
Finally, it is worth mentioning that the applications of the memory
analysis (points 1 and 2) go beyond the detection of needless write
accesses: a precise model of the EVM memory is crucial to enhance the
accuracy of any posterior analysis (see, e.g., \cite{LagouvardosGTS20}
for other concrete applications of a memory analysis).
  





\secbeg
\section{Memory  Layout and Motivating Examples}\label{sec:memory-stor-layo}


\paragraph{Memory Opcodes. } The EVM instruction set contains
the usual instructions to access memory: the most basic
instructions that operate on memory are \code{MLOAD} and
\code{MSTORE}, which load and store a 32-byte word from
memory, respectively.\footnote{Although the local memory is
  byte addressable with instruction \code{MSTORE8}, to
  keep the description simpler, we only consider the general
  case of word-addressable \code{MSTORE}.}
The \texttt{solc} compiler generates code to handle memory with a
cumulative model
in which
memory is allocated along the execution of the program and is never
released.  In contrast to other bytecode virtual machines, like the
\textit{Java Virtual Machine}, the EVM does not have a particular
instruction to allocate memory. The allocation is performed by a
sequence of instructions that use the value stored at address $0x40$
as the \emph{free memory pointer}, i.e., a pointer to the first memory
address available for allocating new memory.  In what follows, we use
\mem{x} to refer to the content stored in memory at location x.

\paragraph{Memory Slots.} 
In the general case, memory is structured as a sequence of
\textit{slots}. A slot is composed of consecutive memory locations
that are accessed by using its initial memory location, which we call
the \textit{base reference} (\baseref for short) of the slot, plus the
corresponding offset needed to access a specific location within the
slot.  Slots usually store (part of) some data structure created in
the Solidity program (e.g., an array or a struct) and whose length can
be known.


\begin{figure}[t]

\begin{minipage}{4.4cm}

\begin{lstlisting}[name=running, numbersep=6pt, language=Solidity]
$\label{running1:st1}$struct TokenOwnership {
$\label{running1:st2}$  address addr;
$\label{running1:st3}$  uint64 startTs;
$\label{running1:st4}$  bool burned;
$\label{running1:st5}$}

contract Running1 {
  //...
  function unpackedOwnership
    (uint256 packed) public
$\slotcc{1}$$\slotcc{2}$    $\label{running1:returns}$ returns (TokenOwnership memory ownership) {
    $\label{running1:pp1}$ownership.addr = ...;
    $\label{running1:pp2}$ownership.startTs = ...;
    $\label{running1:pp3}$ownership.burned = ...;
  }
}
\end{lstlisting}
\end{minipage}
\begin{minipage}{7.7cm}
\begin{lstlisting}[name=running,firstnumber = 17, numbersep=6pt, language=Solidity]
contract Running2 {
  Running1 c;
  mapping(uint256=>uint256) private _packedOwnerships;
  // ...
  function _ownershipAt(uint256 i) private 
$\slotcc{6}$      $\label{ex2:funchead3}$  returns (TokenOwnership memory) {
$\slotcc{7}$    $\label{ex2:externalcall}$return c.unpackedOwnership(_packedOwnerships[i]);
  }
$\label{ex2:funchead1}$  function explicitOwnershipOf(uint256 tokenId)  
$\slotcc{3}$$\label{ex2:funchead2}$        public returns (TokenOwnership memory) {
$\slotcc{4}$  $\label{ex2:varres1}$  TokenOwnership memory ownership;
$\label{ex2:return1}$$\slotcc{5}$    if (...) { return ownership; }
$\slotcc{8}$    $\label{ex2:callres}$ownership = _ownershipAt(tokenId);
    //...
$\slotcc{5}$    $\label{ex2:return2}$return ownership;
  }
}
\end{lstlisting}
\end{minipage}
                      
% ...
\secbeg\secbeg\secbeg
\caption{\textsf{Excerpt of smart contract ERC721A.}}
\secbeg\secbeg\secbeg
\label{fig:running1}
\end{figure}

\begin{example}[slots]
\label{ex:solcode}
Fig.~\ref{fig:running1} shows an excerpt of smart contract
ERC721A~\cite{running} which contains two different contracts
\code{Running1} and \code{Running2}.
We have omitted non-relevant instructions such as those that appear at
lines~\ref{running1:pp1}-\ref{running1:pp3}
(L\ref{running1:pp1}-L\ref{running1:pp3} for short).
The contract \code{Running1} to the left of Fig.~\ref{fig:running1}
contains the public function {\code{unpackedOwnership}} that returns a
struct of type \code{TokenOwnership} defined at
L\ref{running1:st1}-L\ref{running1:st5}.
The contract \code{Running2}, shown to the right, contains the public
function \code{explicitOwnershipOf} that returns, depending on a
non-relevant condition, an empty struct of type \code{TokenOwnership}
(L\ref{ex2:return1}) or the \code{TokenOwnership} received from a call
to function \code{unpackedOwnership} of contract \code{Running1}
(L\ref{ex2:externalcall}), which is done in the private function
\code{\_ownershipAt}.  The execution of function
{\code{unpackedOwnership}} in \code{Running1} allocates two different
memory slots at L\ref{running1:returns}: \slot{1}, for the returned
variable \lstcode{ownership}, and \slot{2}, which is used for actually
returning from the function the contents of \lstcode{ownership}:
\vspace{-0.2cm}
  \begin{center}

\begin{tikzpicture}[->, start chain, thick]
  \tikzset{
      onelement/.style={
          draw, rectangle,minimum size=14pt, 
          inner sep=0pt, text=black, node distance=-1pt
      }
  }

  \tikzset{
    slot/.style={
      draw, rectangle,minimum height=14pt, minimum width=60pt, 
      inner sep=0pt, text=black, node distance=-1pt
    }
}

  \node[onelement] (m0)  [label=below:$\evmcode{0x00}$] {};
  \node[onelement] (m20) [right=of m0, label=below:$\evmcode{0x20}$] {};
  \node[onelement] (m40) [right=of m20, label=below:$\evmcode{0x40}$] {};
  \node[onelement] (m60) [right=of m40, label=below:$\evmcode{0x60}$] {};
  \node[slot] (m80) [right=of m60, label=below:$\evmcode{bref=0x80}$, label=above:$\slot{1} \nodetext{(L\ref{running1:returns})}$] {\scriptcode{ownership}};
  \node[slot] (slot2) [right=of m80, label=below:$\evmcode{bref=0x80+0x60}$, label=above:$\slot{2} \nodetext{(L\ref{running1:returns})}$] {\scriptcode{return}};
  
  \draw[->,dashed] (m40.north) -- +(0,.3) -| (m80.north west);
  \draw[->] (m40.north) -- +(0,.5) -| (slot2.north west);

\end{tikzpicture}

  \end{center}
\vspace{-0.35cm}
  \noindent
  The function \code{explicitOwnershipOf} in \code{Running2}
  makes a more intensive use of the memory which 
  can be seen in this graphical representation:
  \vspace{-0.2cm}
  \begin{center}

  \begin{tikzpicture}[->, start chain, thick]
    \tikzset{
      onelement/.style={
          draw, rectangle,minimum height=14pt, minimum width=9pt,
          inner sep=0pt, text=black, node distance=-1pt
      }
  }
  
    \tikzset{
      slot/.style={
        draw, rectangle,minimum height=14pt, minimum width=40pt, 
        inner sep=0pt, text=black, node distance=-1pt
      }
    }

    \tikzset{
      slotsep/.style={
        draw, rectangle,minimum height=14pt, minimum width=40pt, 
        inner sep=0pt, text=black, node distance=3pt and 20pt
      }
    }
    
    \node[onelement] (m0)  [] {};
    \node[onelement] (m20) [right=of m0, label=above:$~~~\evmcode{0x00-0x60}$] {};
    \node[onelement] (m40) [right=of m20] {};
    \node[onelement] (m60) [right=of m40] {};
      \node[slot] (slot3) [right=of m60, label=above:$\slot{3} \nodetext{(L\ref{ex2:funchead2})}$] 
            {\nodetext{returns} };
    
            \node[slot] (slot4) [right=of slot3, label=above:$\slot{4} \nodetext{(L\ref{ex2:varres1})}$] 
            {\scriptcode{ownership}};
    
  
    \node[slotsep] (slot6) [above right=of slot4,label=above:$\slot{6} \nodetext{(L\ref{ex2:funchead3})}$] 
            {\nodetext{returns}};

    \node[slot] (slot7) [right=of slot6,label=above:$\slot{7} \nodetext{(L\ref{ex2:externalcall})}$] 
            {\scriptcode{return} };

    \node[slot] (slot8) [right=of slot7,label=above:$\slot{8} \nodetext{(L\ref{ex2:callres})}$] 
            {\nodetext{call res.}};

    \node[slotsep] (slot5) [below right=of slot8,label=above:$\slot{5} \nodetext{(L\ref{ex2:return2}-L\ref{ex2:return1})}$] 
            {\scriptcode{return}};

    \draw[->] (slot4.east) -- (slot5.west);
    \draw[->] (slot8.east) -- (slot5.north west);
    \draw[->] (slot4.east) -- (slot6.south west);

  \end{tikzpicture}

  \end{center}
  \vspace{-0.2cm}
  \noindent
  The execution of this function might create up to six different
  slots. At L\ref{ex2:funchead2} and L\ref{ex2:varres1}, it creates
  two slots, one for the struct declared in the \mcode{returns} part
  of the function header (\slot{3}) and one for the local variable
  \code{ownership} (\slot{4}). Depending on the evaluation of the
  condition in the \mcode{if} sentence, it might create the slots
  needed to perform the call to \code{\_ownershipAt} and,
  consequently, the external call to
  \code{Running1.unpackedOwnership}. The invocation to the private
  function involves three slots: one for the struct declared in the
  \mcode{returns} part of \code{\_ownershipAt} in L\ref{ex2:callres}
  (\slot{6}), one slot to manage the external call data in
  L\ref{ex2:externalcall} (\slot{7}), and one slot for storing the
  results of the private function \code{\_ownershipAt} in
  L\ref{ex2:callres} (\slot{8}). Finally, a new slot (\slot{5}) is
  created for returning the results of
  \code{explicitOwnershipOf}. This new slot might contain the contents
  of \slot{4} or \slot{8}, depending on the \mcode{if} evaluation.
\end{example}

When an amount of memory $t$ is to be allocated, the slot reservation
is made by reading and incrementing the free memory pointer
(\mem{0x40}) $t$ positions. From this update on, the \emph{base
  reference} to the slot just allocated is used,
and subsequent accesses to the slot are performed by means of this
baseref, possibly incremented by an offset.

\begin{example}[memory slot reservation]
\label{ex:slotres}
The following excerpt of \evm code allocates a slot of type
\code{TokenOwnership}. The EVM bytecode performs three steps: 

\noindent
\begin{minipage}{0.57\textwidth}
  (i) load the current value of the free memory pointer \mem{0x40} that will be
  used as the \baseref of the new slot; (ii) compute
  the new free memory address by adding $t$ to the \baseref; and (iii), store the
  new free memory pointer in \mem{0x40}. Additionally, in the same block of the
  CFG, the slot reservation is followed by the slot initialization at \code{0x19A},
  \code{0x1AB} and \code{0x1B4}. 
\end{minipage}
\begin{minipage}{0.42\textwidth}
  {\begin{lstlisting}[language=evm,escapechar=@,basicstyle=\scriptsize\ttfamily]
  @{\tiny0x175:\hspace{-0.1cm}}@ JUMPDEST
  @{\tiny0x176:\hspace{-0.1cm}}@ PUSH1 0x40
  @{\tiny0x178:\hspace{-0.1cm}}@ MLOAD    // (i) baseref
      DUP1
      PUSH1 0x60 // Sizeof "t"
      ADD  // (ii) baseref+0x60
  @{\tiny0x17D:\hspace{-0.1cm}}@ PUSH1 0x40
  @{\tiny0x17F:\hspace{-0.1cm}}@ MSTORE   // (iii) @\vspace{-0.12cm}@
      $\dots$
  @{\tiny0x19A:\hspace{-0.1cm}}@ MSTORE   // baseref+0x00 @\vspace{-0.12cm}@ 
      $\dots$
  @{\tiny0x1AB:\hspace{-0.1cm}}@ MSTORE   // baseref+0x20 @\vspace{-0.12cm}@
      $\dots$
  @{\tiny0x1B4:\hspace{-0.1cm}}@ MSTORE   // baseref+0x40 @\vspace{-0.12cm}@

\end{lstlisting}}
\end{minipage}

\secbeg

\end{example}

Solidity reference type values such as arrays, struct typed variables
and strings are stored in memory using this general pattern, with some
minor differences. However, there are some cases in which the steps
detailed above vary 
 and the size of the slot is not
 known in advance, and thus the free memory pointer cannot be
   updated at this point. 
  For instance, when data is returned
by an external call, its length is unknown beforehand and hence the
free memory pointer is updated only after the memory pointed to 
is written. In other cases, the free memory is used as a temporary
region with a short lifetime, as in the case of parameter passing to
external calls, and the free memory pointer is not updated. These
variants of the general schema must be detected by a precise memory
analysis. To this end, we consider that a slot is in \emph{transient}
  state when its baseref has been read from \mem{0x40} but the
  free memory pointer has not been updated, and it is in
  \emph{permanent} state when the free memory pointer has been pushed
  forward.


\begin{example}[transient slot]
  \label{ex:transient}
  Now we focus on the external call in L\ref{ex2:externalcall} of
  \code{Running2}, which performs a \code{STATICCALL}, reading from
  the stack (see~\cite{yellow} for details) the memory location of the
  input arguments and the location where the results of the call will
  be saved.  Interestingly, both locations reuse the same slot (it
  corresponds to $s_7$) 
  as it can be seen in the following \evm bytecode from
  \code{_ownerShipAt}:

\noindent
\begin{minipage}{6.3cm}
  {\begin{lstlisting}[language=evm,escapechar=@]
    PUSH4 0xb04dd20b // func. selector @\vspace{-0.12cm}@ 
    $\dots$
    PUSH1 0x40
@{\tiny0x114:\hspace{-0.1cm}}@ MLOAD   // baseref transient slot  @\vspace{-0.12cm}@ 
    $\dots$ 
    DUP2
    MSTORE  // stores func. selector
    PUSH1 0x04
    ADD     // offset of funct. args. 
    $\cdots$     // copy func. args.
    MSTORE  // stores func. args. @\vspace{-0.12cm}@ 
    $\dots$
\end{lstlisting}}
\end{minipage}
\begin{minipage}{6cm}
  {\begin{lstlisting}[language=evm,escapechar=@]
    PUSH 0x40
@{\tiny0x132:\hspace{-0.1cm}}@ MLOAD      // slot baseref @\vspace{-0.12cm}@ 
    $\dots$
@{\tiny0x139:\hspace{-0.1cm}}@ STATICCALL // external call @\vspace{-0.12cm}@ 
    $\dots$
    PUSH1 0x40
@{\tiny0x151:\hspace{-0.1cm}}@ MLOAD     // slot baseref
    RETURNDATASIZE   @\vspace{-0.12cm}@ 
    $\dots$
    ADD       // baseref + data size @\vspace{-0.12cm}@ 
    $\dots$
@{\tiny0x15E:\hspace{-0.1cm}}@ PUSH1 0x40
@{\tiny0x160:\hspace{-0.1cm}}@ MSTORE   // permanent slot 
\end{lstlisting}}
\end{minipage}

\secbeg
\noindent
The call starts by reading the free memory pointer (at \code{0x114})
and storing at that address the arguments' data (which include the
function selector as first argument).
Importantly, the pointer is not pushed forward when the input arguments
are written and thus the slot remains in transient state.
%
%
Once the call at \code{0x139} is executed, the result is written to
memory from the baseref on (overwriting the locations used for the
input arguments) and the slot is finally made permanent by reading the
free memory pointer again (\code{0x151}) and updating it
(\code{0x160}) by adding the actual return data size
(\code{RETURNDATASIZE}).

  Transient slots are also used when returning data from a public
  function to an external caller. In that case, the \evm code of the
  public function halts its execution using a \mcode{RETURN}
  instruction. It reads from the stack the memory location where the
  length and the data to be returned are located. However, it does not
  change \mem{0x40} because the function code halts its execution at
  this point, as we can see in the \EVM code of
  {\code{explicitOwnershipOf}} (corresponds to slot $s_5$):
  

  \noindent
\begin{minipage}{6.5cm}
\begin{lstlisting}[language=evm]
   PUSH1 0x40
@{\tiny0x4D:\hspace{-0.15cm}}@ MLOAD   //@\textbf{{\color{red}ret}}@ slot baseref @\vspace{-0.12cm}@
   $\dots$
   MSTORE  // @\textbf{{\color{red}ret}}@.addr (@\textbf{{\color{red}ret}}@+0x00) @\vspace{-0.12cm}@
   $\dots$
   MSTORE  // @\textbf{{\color{red}ret}}@.startTs (@\textbf{{\color{red}ret}}@+0x20) @\vspace{-0.12cm}@
   $\dots$
   MSTORE  // @\textbf{{\color{red}ret}}@.burned (@\textbf{{\color{red}ret}}@+0x40) @\vspace{-0.12cm}@
\end{lstlisting}
\end{minipage}
\begin{minipage}{6cm}
\begin{lstlisting}[language=evm]
   $\dots$
   PUSH1 0x40
@{\tiny0x5A:\hspace{-0.15cm}}@ MLOAD  //@\textbf{{\color{red}ret}}@ slot revisit
   DUP1     
   SWAP2  //Baseref of @\textbf{{\color{red}ret}}@ plus size
   SUB    //Size of @\textbf{{\color{red}ret}}@ data 
@{\tiny0x5E:\hspace{-0.15cm}}@ SWAP1
@{\tiny0x5F:\hspace{-0.15cm}}@ RETURN //@\textbf{{\color{red}ret}}@ returned
\end{lstlisting} 
\end{minipage}

\noindent
The baseref for the return slot is read (at \code{0x4D}) and it is
used as a transient slot to write the struct contents to be returned
by adding the corresponding offset for each field contained in the
struct (instructions on the left column). The code on the left ends
with the baseref plus the size of the stored data on top of the stack.
After that, the baseref is read again (top of the right column) and
the length of the returned data is computed (by subtracting the
baseref to the baseref plus the size of the stored data) before
calling the \mcode{RETURN} instruction.
\end{example}


\secbeg

\secbeg
\section{Inference of Needless Write Accesses}
\label{sec:mem-analysis}
\secbeg

This section presents our static  inference of needless write accesses.
We first provide some background in
Sec.~\ref{sec:background} on the type of control-flow-graph (CFG) and
static analysis we rely upon.  Then, the analysis is divided into
three consecutive steps:
(1) the slot analysis, which is introduced in
Sec.~\ref{sec:allocation-analysis}, to identify the slots created
along the execution and the program points at which they are
allocated;
(2) the slot usage analysis, presented in
Sec.~\ref{sec:slots-analysis}, which computes the read and write
accesses to the different slots identified in the previous step; and
(3) the detection of needless write accesses, given in
Sec.~\ref{sec:rw-analysis}, which finds those program points where
there is a write access to a slot which has no read access later on.

\subsection{Context-Sensitive CFG and Flow-Sensitive Static Analysis}
\label{sec:background}

The construction of the CFG of Ethereum smart contracts is a key part
of any decompiler and static analysis tool and has been subject of
previous research~\cite{GrechLTS22,gigahorse,SchneidewindGSM20}.
 The more precise the
CFG is, the more accurate our analysis results will be. In particular,
context-sensitivity~\cite{GrechLTS22} on the CFG construction is vital
to achieve precise results. Our implementation of context-sensitivity
is realized by cloning the blocks which are reached from different
contexts.

\begin{example}[context-sensitive CFG]
\label{ex:cloning}
The EVM code of \code{Running2} 
creates
multiple slots for handling structs of type
\code{TokenOwnership}. Interestingly, all these slots are created by
means of the same EVM code shown in Ex.~\ref{ex:slotres}, which
corresponds to the CFG block that starts at program point
\code{0x175}. As this block is reached from different contexts, the
context-sensitive CFG contains three clones of this block:
\code{0x175}, which creates \slot{3} at L\ref{ex2:funchead2};
\code{0x175\_0}, which creates \slot{4} used at L\ref{ex2:varres1};
and \code{0x175\_1}, which reserves \slot{6}, created at
L\ref{ex2:funchead3}.
Block cloning means that program points are cloned as well, and we
adopt the same subindex notation to refer to the program points
included in the cloned block: e.g. program point \code{0x178} contains
the \code{MLOAD 0x40} that gets the baseref of the slot
reserved at block \code{0x178}, and \code{0x178_0} to the same \code{MLOAD}
but at \code{0x178_0}, etc.
\end{example}

In what follows, we assume that cloning has been made and
the memory analysis using the resulting CFG (with clones) is thus
context-sensitive as well, without requiring additional extensions.
As usual in
standard analyses~\cite{DBLP:books/daglib/0098888}, one has to define
the notion of \emph{abstract state} which defines the abstract
information gathered in the analysis and the \emph{transfer function}
which models the analysis output for each possible input.  Besides
context-sensitivity, the two analyses that we will present in the next
two sections are \emph{flow-sensitive}, i.e., they make a
flow-sensitive traversal of the CFG of the program using as input for
analyzing each block of the CFG the information inferred for its
callers.  When the analysis reaches a CFG block with new information,
we use the operation $\lub$ to join the two abstract states, and the
operator $\sqsubseteq$  to detect that a fixpoint is reached
and, thus, that the analysis terminates.  The operations $\lub$ and $\sqsubseteq$,
the abstract state, and transfer function, will be defined for each
particular analysis.

\secbeg
\subsection{Slot Analysis}
\label{sec:allocation-analysis}

The slot analysis aims at inferring the \emph{abstract slots}, which
are an abstraction of all memory allocations that will be made along
the program execution. The slots inferred are \emph{abstract} because
over-approximation is made at the level of the program points at which
slots are allocated. Therefore, an abstract slot might represent
multiple (not necessarily consecutive) real memory slots, e.g., when
memory is allocated within a loop.  The slot analysis will look for
those program points at which the value stored in \mem{0x40} is read
for reserving memory space. These program points are relevant in the
analysis for two reasons: firstly, to obtain the baseref of the memory
slot, and, secondly, because from this point on, the memory
reservation of the corresponding slot has started and it is pending to
become permanent at some subsequent program point.  The output of the
slot analysis is a set which contains the allocated abstract slots,
named $\closed$ in Def.~\ref{sec:slot-analysis} below.  Each allocated
abstract slot (i.e., each element in $\closed$) is in turn a set of
program points, as the same abstract slot might have several program
points where \mem{0x40} is read before its reservation becomes
permanent. In order to obtain $\closed$, the memory analysis makes a
flow-sensitive traversal of the (context-sensitive) CFG of the program
that keeps at every program point the set of transient slots
(i.e. whose baseref has been read but it has not yet made permanent)
and applies the transfer function in Def.~\ref{sec:slots-analysis-1}
to each bytecode instruction within the blocks until a fixpoint is
reached.  An \emph{abstract state} of the analysis is a set
$\opened \subseteq \wp(\ppmload)$, where $\ppmload$ is the set of all
program points at which \mem{0x40} is read.  The analysis of the
program starts with
$\opened = \set{\emptyset}$ at all program points and takes $\lub$
and $\sqsubseteq$ as the set union and inclusion operations.
Termination is trivially guaranteed as the number of program points is finite and so is $\wp(\ppmload)$. In what follows, $\textit{Ins}$ is the set of \evm instructions and, for simplicity, we consider \mcode{MLOAD 0x40} and \mcode{MSTORE 0x40} as single instructions in $\textit{Ins}$.


\begin{definition}[slot analysis transfer function]\label{sec:slots-analysis-1}
Given a program point $pp$ 

\noindent
\begin{minipage}{0.62\textwidth} 
with an instruction $I\in\textit{Ins}$,
an abstract state $\opened$,
and 
$\closing = \{\mcode{MSTORE 0x40}, \mcode{RETURN}, \mcode{REVERT},$ $\mcode{STOP}, \mcode{SELFDESTRUCT}
\}$,
the \emph{slot analysis transfer function} $\nu$
is defined as a mapping $\nu : \textit{Ins} \times \wp(\opened) \mapsto  \wp(\opened)$ computed according to 
the following table:
\end{minipage}
\hspace{0.2cm}
\begin{minipage}{0.38\textwidth}
{
\renewcommand{\arraystretch}{1.5}%
\scriptsize
$\begin{array}{r|l|c|}
        \cline{2-3}
      \cline{2-3}
          \multicolumn{1}{c|}{}&
        \multicolumn{1}{c|}{\bf I} &
                                     \multicolumn{1}{c|}{\bf \nu(I,\opened)}
      \\\cline{2-3}\cline{2-3}
        (1) & 
        \smcode{MLOAD 0x40} &
          \{s \cup \{pp\} ~|~ s \in \opened\} 

        \\\cline{2-3}
           (2) & 
           I \in \closing  &
        \set{\emptyset} 

      \\\cline{2-3}
     (3) & 
     \text{otherwise} 
      &
       \opened
     \\\cline{2-3}
    \end{array}
$
} 
\end{minipage}
\end{definition}

Let us explain intuitively how the above transfer function works. As we have
seen in Sec.~\ref{sec:memory-stor-layo}, in an EVM program all memory
reservations start by reading \mem{0x40} by means of a
\code{MLOAD} instruction preceded by a \code{PUSH 0x40} instruction (case 1 in
Def.~\ref{sec:slots-analysis-1}). In this case, the transfer function
  adds to all sets in $\opened$ the current program point, since this is,
  in principle, an access to the same slots that were already open at
  this program point and are not permanent yet. 
To properly identify 
the slots, our analysis also searches for those program points at which slots
reservations are made permanent (case 2 in Def.~\ref{sec:slots-analysis-1}),
i.e., those program points with instructions $I \in \closing$.
The most frequently used instruction to make a slot reservation
\textit{permanent} is a write access to \mem{0x40} using \code{MSTORE}, that
pushes forward the free memory pointer such that any subsequent read access to
\mem{0x40} will allocate a different slot. The rest of instructions in
$\closing$ finalize the execution in different forms (a normal return, a forced
stop, a revert execution, etc.).   In all such cases, the slot needs to be 
considered as a permanent slot so that we can reason later on potential needless
write accesses involved in it.   The set $\opened$ is empty after these instructions
since all transient (abstract) slots are made permanent after them.
We use the notation $\opened_{pp}$ to refer to
the abstract state computed at program point $pp$.

\begin{example}[slot analysis]
\label{ex:slotsanalysis}
The slot analysis of \code{Running2} 
 starts with
    {\small $\opened_{pp} {=} \set{\emptyset}$} at all program points. When it
    reaches the block that starts at \code{0x175} (see Ex.~\ref{ex:slotres})
    {\small $\opened_{\evmcode{0x175}}$} is {\small $\set{\emptyset}$}
    and it remains empty 
    until \code{0x178}, where the baseref of \slot{3} is read and hence
    {\small $\opened_{\evmcode{0x178}} {=} \set{\set{\mcode{0x178}}}$}. This slot is made
    permanent when the free memory pointer is updated at \code{0x17F}, thus
    having {\small $\opened_{\evmcode{0x17D}} {=}
      \set{\set{\mcode{0x178}}}$} and 
    {\small $\opened_{\evmcode{0x17F}} {=} \set{\emptyset}$}. 
    Following the same pattern, \slot{4} and \slot{6} are resp.
    reserved at instructions \code{0x178\_0} and \code{0x178\_1} and
    closed at \code{0x17F\_0} and \code{0x17F\_1} (at the cloned
    blocks).
    On the other hand, the baseref of \slot{5} is read at two
    consecutive program points (\code{0x4D} and \code{0x5A}) and
    updated at \code{0x5F}, and thus, we have {\small
      $\opened_{\evmcode{0x4D}} {=} \set{\set{\mcode{0x4D}}}$} and the
    same until {\small
      $\opened_{\evmcode{0x5A}} {=}
      \set{\set{\mcode{0x4D},\mcode{0x5A}}}$} and again the same until
    {\small $\opened_{\evmcode{0x5F}} {=} \set{\emptyset}$}. 
    Finally, after the execution of \code{STATICCALL} (see
    Ex.~\ref{ex:transient}) we have three consecutive reads of
    \mem{0x40} at \mcode{0x114}, \mcode{0x132} and \mcode{0x151} that
    refer to the same slot \slot{7}, which is made permanent at
    \mcode{0x160}. Therefore, we have {\small
      $\opened_{\evmcode{0x151}} {=} \set{\set{\mcode{0x114},
          \mcode{0x132},\mcode{0x151}}}$} and {\small
      $\opened_{\evmcode{0x160}} = \set{\emptyset}$}.  
\end{example}

Using the transfer function, as mentioned in
Sec.~\ref{sec:background}, our analysis makes a flow-sensitive traversal
of the (context-sensitive) CFG of the program that uses as input for analyzing each block 
the information inferred for its callers. 
When a fixpoint is reached, we have an abstract state for each program point that
we use to compute the set of abstract slots allocated in the program,
named $\closed$.

\secbeg
\begin{definition}\label{sec:slot-analysis}
  The set of allocated abstract slots $\closed$ is defined as \\$\closed = \bigcup_{pp \in \ppmstore} \opened_{pp-1}$, 
  where $\ppmstore$ is the set of all program points $pp{:}I$ where
  $I{\in} \closing$.
\end{definition}
%
\secbeg\secbeg
\begin{example}[$\closed$ computation]
  With the values of $\opened_{\evmcode{0x17F-1}}$, $\opened_{\evmcode{0x17F\_0-1}}$, $\opened_{\evmcode{0x17F\_1-1}}$, $\opened_{\evmcode{0x160-1}}$
    and $\opened_{\evmcode{0x5F-1}}$ from Ex.~\ref{ex:slotsanalysis}, at the end
    of the slot analysis of \code{Running2}, we have: \\
  \noindent
  {\small
  $\closed {=} \set{
  \underbrace{\set{\mcode{0x178}}}_{\slot{3}}, 
  \underbrace{\set{\mcode{0x178\_0}}}_{\slot{4}}, 
  \underbrace{\set{\mcode{0x178\_1}}}_{\slot{6}}, 
  \underbrace{\set{\mcode{0x114}, \mcode{0x132},\mcode{0x151}}}_{\slot{7}},
  \underbrace{\set{\mcode{0x5A}, \mcode{0x4D}}}_{\slot{5}},
  \dots}. 
  $
  }
  \\\noindent 
  Note that, the cloning of block \code{0x175} 
   allows our analysis
  to detect three different slots, $\slot{3}$, \slot{4} and \slot{6}, for the same 
  program point, \code{0x178}, in the original EVM code.
  \end{example}

%




The next example shows the behavior of the analysis when the program
contains loops, and an abstraction is needed for approximating the
slots.






\begin{figure}[t]

\begin{center}
\begin{tabular}{l}
{  \begin{lstlisting}[name=code, numbersep=6pt, firstnumber = 37,language=Solidity]
contract Running3 {
  Running2 c;
  //...
$\slotcc{9}$  $\label{ex3:param}$function explicitOwnershipsOf(uint256[] memory tokenIds)
    $\label{ex3:retparam}$public view returns (TokenOwnership[] memory) {
    unchecked {
      uint256 tokenIdsLength = tokenIds.length;
$\slotcc{10}$$\slotcc{11}$      $\label{ex3:arraydef}$TokenOwnership[] memory ownerships = new TokenOwnership[](tokenIdsLength);
      for (uint256 i; i != tokenIdsLength; ++i) {
$\slotcc{12}$$\slotcc{13}$          $\label{ex3:extcall}$  ownerships[i] = c.explicitOwnershipOf(tokenIds[i]);
      }
$\slotcc{14}$      $\label{ex3:return}$return ownerships;
    }
  }
}
  \end{lstlisting}
  }
\end{tabular}

\secbeg\secbeg\secbeg
\begin{tikzpicture}[->, start chain, thick]
  \tikzset{
    onelement/.style={
        draw, rectangle,minimum height=14pt, minimum width=9pt,
        inner sep=0pt, text=black, node distance=-1pt
    }
}

  \tikzset{
    slot/.style={
      draw, rectangle,minimum height=14pt, minimum width=40pt, 
      inner sep=0pt, text=black, node distance=-1pt
    }
  }

  \tikzset{
    slotarray/.style={
      draw, rectangle,minimum height=14pt, minimum width=17pt, 
      inner sep=0pt, text=black, node distance=-1pt
    }
  }

  \node[onelement] (m0)  [] {};
  \node[onelement] (m20) [right=of m0, label=below:$~~~\evmcode{0x00-0x60}$] {};
  \node[onelement] (m40) [right=of m20] {};
  \node[onelement] (m60) [right=of m40] {};
  \node[slot] (slot9) [right=of m60, label=above:$\slot{9} \nodetext{(L\ref{ex3:param})}$] 
        {\scriptcode{tokenIds}};

  \node[slot] (slot10) [right=of slot9, label=above:$\slot{10} \nodetext{(L\ref{ex3:arraydef})}$] 
        {\scriptcode{ownerships}};

  \draw [-,decorate, decoration = {calligraphic brace, raise = 2pt}] (3.85,0.25) --  (5.65,0.25);
  \node[slotarray] (slot110) [right=of slot10] 
        {\scriptcode{o[0]}};
  \node[slotarray] (slot111) [right=of slot110,label={[label distance=3pt]above:$\slot{11} \nodetext{(L\ref{ex3:arraydef})}$} ] 
        {...};
  \node[slotarray] (slot11n) [right=of slot111] 
        {\scriptcode{o[n]}};

  \node[slotarray] (slot120) [right=of slot11n,label={above:$\slot{12}$}] 
        {\nodetext{c[0]}};

  \node[slotarray] (slot130) [right=of slot120,label={above:$\slot{13}$}] 
        {\scriptcode{o[0]}};

  \node[slotarray] (slot121) [right=of slot130,label={above:$\slot{12}$}] 
        {...};
  \node[slotarray] (slot131) [right=of slot121,label={above:$\slot{13}$}] 
        {...};
  \node[slotarray] (slot12n) [right=of slot131,label={above:$\slot{12}$}] 
        {\nodetext{c[0]}};
  \node[slotarray] (slot13n) [right=of slot12n,label={above:$\slot{13}$}] 
        {\scriptcode{o[n]}};

  \draw [-,decorate, decoration = {calligraphic brace, raise = 2pt}] (9.2,0.25) --  (11.55,0.25);
  \node[slot,slotarray] (slot140l) [right=of slot13n] 
        {\nodetext{r.l}};
  \node[slot,slotarray] (slot140) [right=of slot140l] 
        {\nodetext{r[0]}};
  \node[slot,slotarray] (slot141) [right=of slot140,label={[label distance=3pt]above:$\slot{14} \nodetext{(L\ref{ex3:return})}$}] 
        {...};
  \node[slot,slotarray] (slot14n) [right=of slot141] 
        {\nodetext{r[n]}};

  \draw[->,dashed] (slot10.south) -- +(0,-.2) -| (slot110.south);
  \draw[->,dashed] (slot10.south) -- +(0,-.2) -| (slot111.south);
  \draw[->,dashed] (slot10.south) -- +(0,-.2) -| (slot11n.south);

  \draw[->] (slot10.south) -- +(0,-.4) -| (slot130.south);
  \draw[->] (slot10.south) -- +(0,-.4) -| (slot131.south);
  \draw[->] (slot10.south) -- +(0,-.4) -| (slot13n.south);

\end{tikzpicture}
\end{center}
\secbeg\secbeg\secbeg
\caption{\textsf{Solidity} code of contract \code{Caller}.}
\secbeg\secbeg\secbeg
\label{fig:running-loop}
\end{figure}
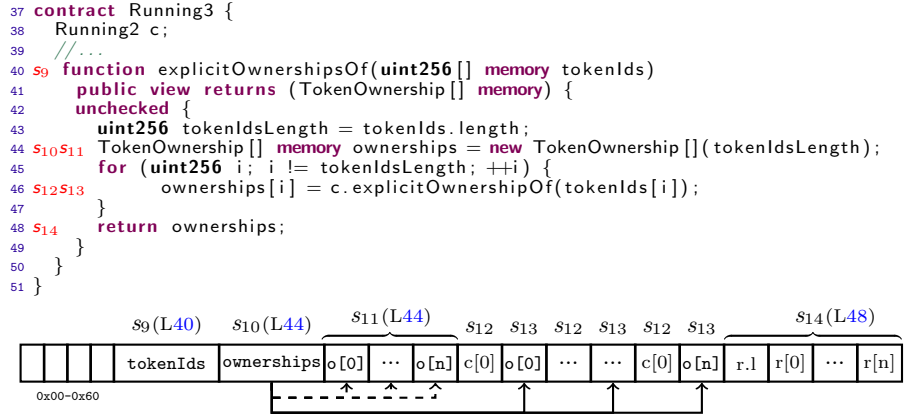

\begin{example}[loops]\label{ex:loops}
  Fig.~\ref{fig:running-loop} shows the contract \code{Running3} that
  includes the function \code{explicitOwnershipsOf} from the smart
  contract at~\cite{running} (made through a \code{STATICCALL}). This function receives an array of token
  identifiers as argument and returns an array of
  \code{TokenOwnership} structs that is populated invoking the
  function \code{explicitOwnershipOf} from \code{Running2} inside a
  loop.
  The slots identified by the analysis for contract \code{Running3}
  shown in Fig.~\ref{fig:running-loop} are:
  \slot{9}, which is created for making a copy of parameter
  \code{tokenIds} to memory;
  \slot{10}, which creates the local array \code{ownerships}
  (L\ref{ex3:arraydef}) that contains the array length and pointers to the structs 
  identified initially by \slot{11} (and later on by \slot{13});
  \slot{12} for \code{STATICCALL} input arguments and return
  data (L\ref{ex3:extcall}); \slot{13} which abstracts the structs for
  storing the \code{STATICCALL} output results
  (L\ref{ex3:extcall}); and \slot{14}, which includes the length of ownership and a copy of \slot{13}
  for returning the results (L\ref{ex3:return}).
  The important point is that, the local array declaration at
  L\ref{ex3:arraydef} produces a loop to allocate as many structs as
  elements are contained in the array. For this reason, \slot{11} is
  an abstract slot that represents all \code{TokenOwnership}'s
  initially added to the array. Similarly, \slot{12} and \slot{13} are
  created inside the \mcode{for} loop, and each abstract slot
  represents as many concrete slots as iterations are performed by the
  loop. Note that, each iteration of the loop creates one instance of
  \slot{12} for getting the results from the call, and it is copied
  later to \slot{13} and pointed by \code{ownerships} (\slot{10}).
\end{example}

As notation, we will use a unique numeric identifier (1, 2, $\ldots$)
to refer to each abstract slot (represented in $\closed$ as a set) and
retrieve it by means of function {\small $get\_id(a), a\in\closed$}.
We use $\allslots$ to refer to the set of all such identifiers in the
program.
Also, given a program point $pp$ with an instruction \code{MLOAD
  0x40}, we define the function $\getslot(pp)$ to retrieve the
identifiers of the elements of $\closed$ that might be referenced at
$pp$ as follows: {\small
  $\getslot(pp) = \{id ~|~ a \in \closed \wedge pp \in a \wedge id =
  get\_id(a) \}.$}
%
Soundness of the analysis is stated in the next theorem.
\begin{theorem}[soundness of slot analysis]\label{th:slots-analysis}
  The set $\closed$ is a sound over-approximation of the concrete
  slots allocated along a program execution.
\end{theorem}
The proof 
basically amounts to showing that all concrete accesses to the same
slot are approximated by the same abstract slot (see details in Appendix~\ref{ap:proofs}).
. 

\secbeg\secbeg
\subsection{Slot Access Analysis}
\label{sec:slots-analysis}

While Sec.~\ref{sec:allocation-analysis} looked for allocations, the
next step of the analysis is the inference of the program points at
which the inferred abstract slots might be
accessed. To do so, our slot access analysis needs to propagate the
references to the abstract slots that are
saved at the different positions of the execution stack. 
Importantly, we keep track, not only of the stack positions, but also,
in order to abstract complex data structures stored in memory (e.g.,
arrays of structs), we need to keep track of the abstract slots that
could be saved at memory locations. 
As seen in Ex.~\ref{ex:loops}, a memory location within a slot might
contain a pointer to another memory location of another slot, as it
happens when nested data structures are used.  Thus, an abstract state
is a mapping at which we store the potential slots saved at stack
positions or at memory locations within other slots.

\begin{definition}[memory analysis abstract state]
\label{def:memory-analysis-as}
  A \textit{memory analysis abstract state} is a mapping $\pi$ of the
  form $\stackset \cup \allslots \mapsto \wp(\allslots)$.
    
\end{definition}

$\stackset$ is the set containing all stack positions, which we
represent by natural numbers from 0 (bottom of the stack) on, and
$\allslots$ is the set of abstract slots identifiers computed in
Sec.~\ref{sec:allocation-analysis}.  We refer to the set of all memory analysis abstract states as $AS$.
Note that, for each entry, we keep a set of potential slots for each
stack position because a block might be reached from several blocks
with different execution stacks, e.g., in loops or
\textit{if-then-else} structures.
In what follows, we assume that, given a value $k$, the map $\pi$
returns the empty set when $k \not\in dom(\pi)$.
The inference is performed by a flow-sensitive analysis (as described
in Sec.~\ref{sec:background}) that keeps track of the information
about the abstract slots used at any program point by means of the
following transfer function.

\begin{definition}[memory analysis transfer function]
\label{def:transfer-accesses}
Given an instruction $I$ with $n$ input
operands at program point $pp$ and an abstract state $\pi$, the \emph{memory analysis
  transfer function} $\tau$ is defined as a mapping
{\small $\tau : \textit{Ins} \times AS \mapsto AS$} of the form:
\\[-0.2cm]
\secbeg\secbeg
\begin{minipage}{0.55\textwidth}
{
\renewcommand{\arraystretch}{1.5}%
\scriptsize
\hspace*{-0.5cm}
\[
    \begin{array}{r|l|l|}
       \cline{2-3}
        \multicolumn{1}{c|}{}&
        \multicolumn{1}{c|}{I} &
        \multicolumn{1}{c|}{\tau(I,\pi)} 

        \\\cline{2-3}
        (1) & 
        \smcode{MLOAD 0x40} &
        \pi[t \mapsto  \getslot(pp)] 

        \\\cline{2-3}
           (2) & 
     \smcode{MLOAD} &
     \pi[t \mapsto \{m ~|~ s \in \pi(t) \wedge m {\in} \pi(s)\}]

      \\\cline{2-3}
      (3) & 
      \smcode{MSTORE}  &
      \pi [s \mapsto \pi(s)\cup \pi(t{-}1)]\backslash \{t,t{-}1\} ~ \forall s {\in} \pi(t)  
      \\\cline{2-3}
      

    \end{array}
  \]
}
\end{minipage}
\begin{minipage}{0.45\textwidth}
{
\renewcommand{\arraystretch}{1.5}%
\scriptsize
\hspace*{-0.5cm}
\[
    \begin{array}{r|l|l|}
       \cline{2-3}
        \multicolumn{1}{c|}{}&
        \multicolumn{1}{c|}{I} &
        \multicolumn{1}{c|}{\tau(I,\pi)} 

        \\\cline{2-3}


      (4) & 
      \smcode{SWAP}i  &
      \pi [t \mapsto \pi(t-i), t-i \mapsto \pi(t)]
      \\\cline{2-3}
      (5) & 
      \smcode{DUP}i  &
      \pi [t+1 \mapsto \pi(t-i+1)]
      \\\cline{2-3}
     (6) & 
     \text{otherwise} 
      &
      \pi \backslash x  \hspace{1.4cm}{t{-}n < x \leq t} 
     \\\cline{2-3}
    \end{array}
  \]
}
\end{minipage}
 $t{=}top(pp)$ is the numerical position of the top
of the stack before executing $I$.
\end{definition}

Let us explain the above definition. The transfer function
distinguishes between two different types of \code{MLOAD}: (1)
accesses to location \mem{0x40}, which return the baseref of
the slots that might be used, taking them from the previous analysis
through $\getslot(p)$; and (2) other \code{MLOAD} instructions, which
could potentially return slot baserefs from memory locations.
Therefore, we have to consider two possibilities: if we are reading a
memory location which reads a generic value (e.g. a number) then
$\pi(t) = \emptyset$; if we are reading a memory location that might
store an abstract slot, then $\pi(t)$ contains all abstract slots that
might be stored at that memory location.
Regarding (3), \code{MSTORE} has two operands: the operand at $t$ is
the memory address that will be modified by \code{MSTORE}, and the
operand at $t-1$ is the value to be stored in that address.
For each element $s$ in $\pi(t)$, the analysis adds the abstract slots
that are in $\pi(t{-}1)$.
Other instructions that are also treated by the analysis are
\code{SWAP*} and \code{DUP*} shown in (4-5), that exchange or
copy the elements of the stack that take part in the operation.
Finally, all other operations delete the elements of the stack that
are no longer used based on the number of elements taken and written
to the stack (case 6).

%

\begin{figure}[t]
  \secbeg\secbeg\secbeg\secbeg
  {\scriptsize
\[
\begin{array} {lll||lll}
\text{\textsf{PP}} & \text{\textsf{Instr}} & ~~\pi& \text{\textsf{PP}} & \text{\textsf{Instr}} & ~~\pi \\ \hline
\evmcode{0x175\_1}  &\evmcode{JUMPDEST} & \{3 {\mapsto} s_3, 4 {\mapsto} s_4\} &\evmcode{0x19A\_1}  &\evmcode{MSTORE} & \{3 {\mapsto} s_3, 4 {\mapsto} s_4,8 {\mapsto} s_6, 9 {\mapsto} s_6\} \\
\evmcode{0x176\_1}  &\evmcode{PUSH1 0x40} & \{3 {\mapsto} s_3, 4 {\mapsto} s_4 \}&... \\
\evmcode{0x178\_1}  &\evmcode{MLOAD} & \{3 {\mapsto} s_3, 4 {\mapsto} s_4, 8 {\mapsto} s_6\}&\evmcode{0x1A9\_1}  &\evmcode{AND} & \{3 {\mapsto} s_3, 4 {\mapsto} s_4,8 {\mapsto} s_6, 9 {\mapsto} s_6\} \\
\evmcode{0x179\_1}  &\evmcode{DUP1} & \{3 {\mapsto} s_3, 4 {\mapsto} s_4, 8 {\mapsto} s_6, 9 {\mapsto} s_6\}&\evmcode{0x1AA\_1}  &\evmcode{DUP2} & \{3 {\mapsto} s_3, 4 {\mapsto} s_4,8 {\mapsto} s_6, 9 {\mapsto} s_6, 11 {\mapsto} s_6\} \\
\evmcode{0x17A\_1}  &\evmcode{PUSH1 0x60} & \{3 {\mapsto} s_3, 4 {\mapsto} s_4,8 {\mapsto} s_6, 9 {\mapsto} s_6\}&\evmcode{0x1AB\_1}  &\evmcode{MSTORE} & \{3 {\mapsto} s_3, 4 {\mapsto} s_4,8 {\mapsto} s_6, 9 {\mapsto} s_6\} \\
\evmcode{0x17C\_1}  &\evmcode{ADD} & \{3 {\mapsto} s_3, 4 {\mapsto} s_4,8 {\mapsto} s_6, 9 {\mapsto} s_6\} &... \\
\evmcode{0x17D\_1}  &\evmcode{PUSH1 0x40} & \{3 {\mapsto} s_3, 4 {\mapsto} s_4,8 {\mapsto} s_6, 9 {\mapsto} s_6\}&\evmcode{0x1B2\_1}  &\evmcode{ISZERO} & \{3 {\mapsto} s_3, 4 {\mapsto} s_4,8 {\mapsto} s_6, 9 {\mapsto} s_6\} \\
\evmcode{0x17F\_1}  &\evmcode{MSTORE} & \{3 {\mapsto} s_3, 4 {\mapsto} s_4,8 {\mapsto} s_6\}&\evmcode{0x1B3\_1}  &\evmcode{DUP2} & \{3 {\mapsto} s_3, 4 {\mapsto} s_4,8 {\mapsto} s_6, 9 {\mapsto} s_6, 11 {\mapsto} s_6\} \\
\evmcode{0x180\_1}  &\evmcode{DUP1} & \{3 {\mapsto} s_3, 4 {\mapsto} s_4,8 {\mapsto} s_6, 9 {\mapsto} s_6\}&\evmcode{0x1B4\_1}  &\evmcode{MSTORE} & \{3 {\mapsto} s_3, 4 {\mapsto} s_4,8 {\mapsto} s_6, 9 {\mapsto} s_6\} \\
...&&&\evmcode{0x1B5\_1}  &\evmcode{POP} & \{3 {\mapsto} s_3, 4 {\mapsto} s_4, 8 {\mapsto} s_6\} \\

\evmcode{0x198\_1}  &\evmcode{AND} & \{3 {\mapsto} s_3, 4 {\mapsto} s_4,8 {\mapsto} s_6, 9 {\mapsto} s_6\} &\evmcode{0x1B6\_1}  &\evmcode{SWAP1} & \{3 {\mapsto} s_3, 4 {\mapsto} s_4,7 {\mapsto} s_6\} \\
\evmcode{0x199\_1}  &\evmcode{DUP2} & \{3 {\mapsto} s_3, 4 {\mapsto}
                                      s_4,8 {\mapsto} s_6, 9 {\mapsto}
                                      s_6, 11 {\mapsto} s_6\}&\evmcode{0x1B7\_1}  &\evmcode{JUMP} &  \{3 {\mapsto} s_3, 4 {\mapsto} s_4,7 {\mapsto} s_6\} \\
\end{array} 
\]
}
\vspace{-0.5cm}
  \caption{\textsf{Block of the CFG  that reserves memory slot for struct }}
  \label{fig:accesses}
  \secbeg\secbeg\secbeg
\end{figure}

\begin{example}[transfer]
  \label{ex:memory-analysis}
  Now we focus on the analysis of block \code{0x175}, shown in Fig.~\ref{fig:accesses}.
  As we have already explained, this block is responsible for creating the
  memory needed to work with several structs of type \code{TokenOwnership} and
  it is thus cloned in the CFG. In particular, we focus on the
  clone \code{0x175\_1}.
The analysis of the block starts with a stack of size 7 and includes at
positions $3$ and $4$, the abstract slots \slot{3} and \slot{4}, which were
created at L\ref{ex2:funchead1} and L\ref{ex2:funchead2} of
Fig.~\ref{fig:running1}. At \code{0x178\_1}, \mem{0x40} is read, and, by means
of $get\_slots(\mcode{0x178\_1})$ and, considering that {\small$top(\mcode{0x178\_1}) {=}
8$}, we add to $\pi$ a new entry {\small $8 \mapsto s_6$}.
At \code{0x179\_1}, \code{0x180\_1}, \code{0x1AA\_1}, \code{0x1B3\_1} the
transfer function duplicates a slot identifier stored in the stack.
 \code{MSTORE} and \code{POP} instructions of the
example remove a slot identifier from the stack.

\end{example}

As it is flow-sensitive, the analysis of each block of the CFG takes
as input the join $\lub$  of the abstract states computed with the transfer
function for the blocks that jump to it, and keeps applying the memory
analysis transfer function until a fixpoint is reached.  The
operation $A \lub B$ is the result of joining, by means of operation
$\cup$, all entries from maps $A$ and $B$. Operation $\sqsubseteq$ is
defined as expected, $A \sqsubseteq B$, when B includes entries that
are not in $dom(A)$ or when we have an entry $v \in
dom(A) \cap \dom(B)$ such that $A(v) \subseteq B(v)$.  Again,
termination of the computation is guaranteed because the domain is
finite.


\begin{example}[joining abstract states]
  The \EVM code of \code{explicitOwnershipOf} of
   Fig.~\ref{fig:running1} uses $s_5$ in both \mcode{return} sentences at
   L\ref{ex2:return1} and L\ref{ex2:return2} (see Ex.~\ref{ex:solcode}). This
   \EVM code has a single return block which is reachable from two different
   paths from the \mcode{if} statement, and which come with different abstract
   states: (1) the path that corresponds to L\ref{ex2:return1} comes
   with {\small$\pi {=} \{3 \mapsto s_8\}$}, and the other path
   (L\ref{ex2:return2}) with {\small$\pi {=} \{3 \mapsto s_4\}$}. Our
   analysis joins both abstract states resulting in {\small$\pi {=} \{3 \mapsto
   \{s_4,s_8\}\}$}. Because of this join, we get that the \mcode{RETURN}
   instruction that comes from lines L\ref{ex2:return1} and L\ref{ex2:return2}
   might return the content of the slots $s_4$ or $s_8$.
\end{example}

 When
the fixpoint is reached,
the analysis has computed an abstract state for each program point $pp$, denoted by
$\pi_{pp}$ in what follows.

\begin{example}[complex data structures]
  The analysis of the code at Fig.
  \ref{fig:running-loop} shows how it deals with data structures that
  might contain pointers to other structures, e.g. \code{ownerships}. The
  abstract slot that represents variable \code{ownerships} is $s_{10}$, which is
  written, by means of \code{MSTORE} at two program points, say $pp_1$ 
  and $pp_2$ 
  which, resp., come from L\ref{ex3:arraydef} and L\ref{ex3:extcall} of
  the Solidity code.
  The input abstract state that reaches $pp_1$ is 
  {\small $\set{2 \mapsto  s_{9}, 6 \mapsto  s_{10}, 8 \mapsto  s_{10}, 9 \mapsto  s_{11}, 10 \mapsto  s_{10}}$},
  and the transfer function of \code{MSTORE} leaves the abstract state as 
  {\small $\pi_{pp_1} = \set{2 \mapsto  s_{9}, 6 \mapsto  s_{10}, 8 \mapsto  s_{10}, s_{10} \mapsto  s_{11}}$}.
%
%
  \noindent At this point, we can see that variable \code{ownerships} is initialized with
  empty structs 
  and, to represent it, our analysis includes in
  $\pi$ the entry $s_{10} \mapsto s_{11}$ as it is described in instruction
  \code{MSTORE} of the transfer function at
  Def.~\ref{def:transfer-accesses}. 
  The second write to $s_{10}$ is performed by another \code{MSTORE} instruction at
  $pp_2$. 
  The input abstract state for $pp_2$ is  
  {\small $\set{{2} \mapsto  s_{9}, 5 \mapsto  s_{10}, 7 \mapsto  s_{13}, 8 \mapsto  s_{13}, 9 \mapsto  s_{10}, s_{10} \mapsto  s_{11}}$},
  and thus we get
  {\small $\pi_{pp_2} = \set{{2} \mapsto  s_{9}, 5 \mapsto  s_{10}, 7 \mapsto  s_{13}, s_{10} \mapsto  \{s_{11}, s_{13}\}}$}.
  \noindent Interestingly, at $pp_2$, 
  we detect that $s_{11}$ might also store the structs returned by the
  call to \code{c.explicitOwnershipOf(tokenIds[i])}, identified by
  $s_{13}$, which is added to {\small
    $s_{10} \mapsto \{s_{11}, s_{13}\}$}.
  Finally, $s_{10}$ is read at the end of the method, returning the
  set {\small $\{s_{11}, s_{13}\}$}, to copy the content of
  \code{ownerships} to $s_{14}$, the slot used in the return.
\end{example}

Soundness guarantees that memory accesses within the same concrete
slot are approximated by the same abstract slot. It easily follows
from
Theorem~\ref{th:slots-analysis} (see details in Appendix~\ref{ap:proofs}).

\secbeg
\begin{theorem}[soundness of slot access analysis]
$\pi_{pp}$  is a sound \\ over-approximation of the slots accessed at
instruction $pp$.
\end{theorem}
 
\secbeg\secbeg
\subsection{Inference of Needless Write Memory Accesses}
\label{sec:rw-analysis}
\secbeg

With the results of the previous analysis, we can compute the maps
$\readset$ and $\writeset$, which are of the form
$pp \mapsto \wp(\allslots)$ and capture the slots that might be read
or written, resp., at the different program points.  To do so, as
multiple EVM instructions, e.g. \mcode{RETURN}, \code{CALL},
\code{LOG}, \code{CREATE}, ..., might read, or write, memory locations
taking the concrete location from the stack, we define functions
$\memread(I)$ and $\memwrite(I)$ that, given an EVM instruction $I$,
return the position in the stack of the address to be read and written
by $I$, resp. If the instruction does not read/write any memory
position, function {\small $\memread(I) = \bot$}/{\small
  $\memwrite(I) = \bot$}. For example, {\small
  $\memread(\text{\code{MLOAD}}) = 0$} as it reads the top of the
stack and {\small $\memwrite(\text{\code{MLOAD}}) = \bot$}, or {\small
  $\memread(\text{\code{STATICCALL}}) = 2$} and {\small
  $\memwrite(\text{\code{STATICCALL}}) = 4$}.
Now, we define the read/write maps $\readset$/$\writeset$: 

\secbeg
\begin{definition}[memory read/write accesses map]
    Given an \EVM program P, such that {\small $pp \equiv I \in P$} and being {\small $t {=} top(pp)$},  we
    define maps $\readset$ and $\writeset$ as follows: 
\secbeg
{
  \small
    \[ 
    \begin{array}{lrlr}
      \readset(pp) {=} \begin{dcases}\emptyset & \memread(I) = \bot\\[-0.15cm]
        \pi_{pp{-}1}(t{-}\memread(I)){} & otherwise \end{dcases}
&~~~~~
      \writeset(pp) {=} \begin{dcases}\emptyset & \memwrite(I) = \bot \\[-0.15cm] 
        \pi_{pp{-}1}(t{-}\memwrite(I)) & otherwise \end{dcases}
    \end{array}    
  \] 
}
\end{definition}

\secbeg\secbeg
\begin{example}[$\readset$/$\writeset$ maps]
Let us illustrate the computation of {\small $\readset(\mcode{0x139})$} and
{\small $\writeset(\mcode{0x139})$}, which contains the  \code{STATICCALL} of
\code{Running2}. With the analysis information obtained from the analysis we
have that 
{\small $top(\mcode{0x139}) = 16$} and 
{\small $\pi_{\mcode{0x138}} = \{3 \mapsto  s_{3}, 4 \mapsto  s_{4}, 7 \mapsto  s_{6}, 10 \mapsto  s_{7}, 12 \mapsto s_7, 14 \mapsto s_7\}$}, thus 
we get {\small$\readset(\mcode{0x139}) = \{s_{7}\}$} and {\small$\writeset(\mcode{0x139}) = \{s_{7}\}$}, i.e., the 
slot used for managing the input and the output of the external call. 
Analogously, we get that {\small $\readset(\mcode{0x178}) = \{s_3\}$} and
{\small $\writeset(\mcode{0x178}) = \emptyset$}. 
\end{example}

The last step of our analysis consists in searching for write accesses to slots
which will never be read later. To do so, we use the information computed in
$\readset$ and $\writeset$. Given the CFG of the program and two program points
$p$ and $p2$, we define function {\small $\reachable(p,p2)$}, which returns $true$ when
there exists a path in the CFG from $p$ to $p2$. We define the set \emph{write
leaks}  $\writingleaks$ as follows: 

\begin{definition}\label{def:rw} Given an \EVM program and its 
    $\writeset$ and $\readset$, we define  $\writingleaks$ as \\[0.05cm]
    {\small
        $\writingleaks = \{pw{:}s ~|~ pw \in P \wedge s \in \writeset(pw) \wedge 
        \neg \existsread(pw,s)\}$\\[0.05cm]
    }
    where \small{$\existsread(pw,s) \equiv \exists~ pr \in dom(\readset) ~|~ s \in \readset(pr) \wedge \reachable(pw,pr)$}.
\end{definition}

\secbeg
Intuitively, the set $\writingleaks$ contains those write accesses,
taken from $\writeset$, that are never read by subsequent blocks in
the CFG. As both function \emph{reachable} and the sets $\writeset$
and $\readset$ are over-approximations, the computation of
$\writingleaks$ provides us those write accesses that can be safely
removed, as the next example shows.

\begin{example}
  Our analysis detects that at program points \code{0x19A},
  \code{0x1AB} and \code{0x1B4} there are \code{MSTORE} operations
  that are never read in the subsequent blocks of the CFG. Such
  operations correspond to the memory initialization of \slot{3},
  which is performed at L\ref{ex2:funchead2} of the code of
  Fig.~\ref{fig:running1} (see Ex.~\ref{ex:slotres}). Given that these
  write accesses are the only use of the slot, the whole reservation
  can be safely removed.
  Moreover, the analysis detects that program points \code{0x19A\_1},
  \code{0x1AB\_1} and \code{0x1B4\_1}, which correspond to the
  reservation of \slot{6} performed at L\ref{ex2:funchead3}, are
  detected as needless.
  In essence, it means that \slot{3} and \slot{6} are allocated and
  initialized but are never used in the program.
  Note that, all these program points belong to two blocks cloned:
  (\code{0x175} and \code{0x175\_1}).
  However, the three \code{MSTORE} operations of the other clone of
  the same block (\code{0x175\_0}), which correspond to the allocation
  at L\ref{ex2:varres1} are not identified as non-read, as they might
  be used in the return of the function.
  For this, the precision of the context-sensitive CFG is necessary to
  identify these \code{MSTORE} operations as needless. As a result we
  cannot eliminate the block because it is needed in one of the
  clones, but still we can achieve an important optimization on the
  EVM code by removing the unconditional jumps to this block in the
  other two cases that would avoid completely the execution of all
  these instructions (and their corresponding gas consumption
  \cite{yellow}).
\end{example}



\secbeg
As a corollary of the soundness for the analysis of the previous
sections, we have that the inferred write accesses are needless. 

\begin{corollary}[soundness of needless write accesses inference]
  If $pp \in \writingleaks$, there is no execution of the program in
  which the slot accessed at $pp$ is read after executing the write
  instruction at $pp$.
\end{corollary}


\secbeg
\section{Experimental Evaluation}\label{sec:experiments}
\secbeg

This section reports on the results of the experimental evaluation of
our approach, as described in Sec.~\ref{sec:mem-analysis}. All
components of the analysis are implemented in Python, are open-source,
and can be downloaded from \github where detailed instructions for its
installation and usage are
provided\footnote{\url{https://github.com/costa-group/EthIR/tree/memory_optimizer/ethir}}.
We use external components to build the CFGs (as this is not a
contribution of our work). Our analysis tool accepts smart contracts
written in versions of Solidity up to 0.8.17 and bytecode for the
Ethereum Virtual Machine v1.10.25\footnote{The latest versions
  released up to Oct 2022.}. The experiments have been performed on an
AMD Ryzen Threadripper PRO 3995WX 64-cores and \SI{512}{\giga\byte} of
memory, running Debian 5.10.70.
%
%
In order to experimentally evaluate the analysis, we pulled from
\textsf{etherscan.io}~\cite{etherscanSourceCodes} the Ethereum
contracts bound to the last 5,000 open-source verified addresses whose
source code was available on July 14, 2022. 
From those addresses, the code of 2.18\% of them raises a compilation
error from \lst{solc}. For the code bound to the \numAddresses
remaining addresses, the generation of the CFG (which is not a
contribution of this work) timeouts after 120s on \numTimeouts of
them. Removing such failing cases, we have finally
analyzed 
\numContracts smart contracts, as each address and each Solidity file
may contain several contracts in it.
Note that 84.86\% of the contracts are compiled with
the  \lst{solc} version 0.8, presumably with the most advanced
compilation techniques.
The whole dataset used will be
found at the above \github link.

In order to be in a worst-case scenario for us, we run the memory
analysis after executing the \texttt{solc} optimizer, i.e, we analyze
bytecode whose memory usage may have been optimized already by the
optimizer available in \texttt{solc}. This will allow us also to see
if we can achieve further optimization with our approach.
Unfortunately, we have not been able to apply our tool after running
the super-optimizer GASOL~\cite{AlbertGHR22}, because it does not generate the
optimized bytecode but rather it only reports on the gas and/or size
gains for each of the blocks. Nevertheless, a detailed comparison of
the techniques that GASOL applies and ours is given in
Sec.~\ref{sec:concl-relat-work}, where we justify that GASOL will not
find any of our needless accesses.
%
%
From the \numContracts analyzed contracts, the analysis infers
\numAbsSlots abstract memory slots and detects \needlessAccess
needless write memory accesses in \totalTime. These needless accesses
occur within the code bound to 780 different addresses, i.e., 15.95\%
of the analyzed ones. 
The following table contains the most common public functions with
needless memory accesses.

  \begin{center}
{\small
  \begin{tabular}{llllll}
    \hline
    \multicolumn{1}{|l|}{\bf Function} & \multicolumn{1}{l|}{\bf
                                          \#Acc} &
                                                               \multicolumn{1}{l|}{\bf \#C} &  \multicolumn{1}{||l|}{\bf Function} & \multicolumn{1}{l|}{\bf
                                          \#Acc} &
                                                               \multicolumn{1}{l|}{\bf \#C}  \\ \hline

    \multicolumn{1}{|l|}{\lst{transfer}} &
                                                             \multicolumn{1}{l|}{1745}
                                                                             &
                                                                               \multicolumn{1}{l|}{441}
                                                                                                     &      \multicolumn{1}{||l|}{\lst{release}} &
                                                              \multicolumn{1}{l|}{68} & \multicolumn{1}{l|}{24}  \\ \hline
        \multicolumn{1}{|l|}{\lst{transferFrom}} & \multicolumn{1}{l|}{1736} & \multicolumn{1}{l|}{439} &      \multicolumn{1}{||l|}{\lst{poke}} &
                                          \multicolumn{1}{l|}{60} &
                                                                    \multicolumn{1}{l|}{1}  \\ \hline
        \multicolumn{1}{|l|}{\lst{safeTransferFrom}} & \multicolumn{1}{l|}{162} & \multicolumn{1}{l|}{52} &     \multicolumn{1}{||l|}{\lst{withdraw}} &
                                          \multicolumn{1}{l|}{54} &
                                                                    \multicolumn{1}{l|}{32}  \\ \hline

    \multicolumn{1}{|l|}{\lst{reflectionFromToken}} &
                                                       \multicolumn{1}{l|}{105} & \multicolumn{1}{l|}{6} &     \multicolumn{1}{||l|}{\lst{claimDividends}} &
                                          \multicolumn{1}{l|}{54} &
                                                                    \multicolumn{1}{l|}{10}
    \\ \hline
    \multicolumn{1}{|l|}{\lst{onERC1155BatchReceived}} &
                                                          \multicolumn{1}{l|}{97} & \multicolumn{1}{l|}{13} &     \multicolumn{1}{||l|}{\lst{onERC1155Received}} &
                                          \multicolumn{1}{l|}{48} &
                                                                    \multicolumn{1}{l|}{12}
    \\ \hline
\end{tabular}
}
\end{center}
Column \textbf{\#Acc} shows the number of needless accesses identified
in each function and column \textbf{\#C} the number of different
contracts that contain these functions. Some of them such as
\lst{transferFrom}, \lst{transfer}, \lst{reflectionFromToken} or
\lst{withdraw} are functions widely used in the implementation of
contracts based on ERC tokens. A manual inspection of these 10 public
functions has revealed two different sources for needless accesses:
some of them are due to inefficient programming practices, while
others are generated by the compiler and could be improved.
As regards compiler inefficiencies, we detected bytecode that
allocates memory slots that are inaccessible and cannot be used
because the baseref to access them is not maintained in the stack. For
example, when a struct is returned by a function, it always allocates
memory for this data. However, if the return variable is not named in
the header of the function, the compiler allocates memory for this
data although it will never be accessed. If programmers are aware of this
behavior they can avoid such generation of useless memory but, even
better, this memory usage patterns can be changed in the compiler.
For instance, it is reflected in L\ref{ex2:funchead3} and
L\ref{ex2:funchead2} in Fig.~\ref{fig:running1}, where the functions
do not name the return variable. Hence, the compiler allocates memory
for these \emph{anonymous} data structures which are never
used. Similarly, there are various situations involving external calls
in which the compiler creates memory that is never used. When there is
an external call that does not retrieve any result, the compiler
creates two memory slots, one for retrieving the result from the call,
and another one for copying a potential result to a memory variable
that is never used. Finally, the compiler also creates memory that is
never used for low-level plain calls for currency transfer.
Even though the contract code does not use the second result
returned by the low-level call, the compiler generates code for
retrieving it.
All these potential optimizations have been detected by means of our
inference of needless write accesses and will be communicated to the
\texttt{solc} developers.




\section{Conclusions and Related Work}\label{sec:concl-relat-work}

We have proposed a novel memory analysis for Ethereum smart contracts
and have applied it to infer needless write memory accesses. The
application of our implementation over more than 19,000 real smart
contracts has detected some compilation patterns that introduce
needless write accesses and that can be easily changed in the compiler
to generate more efficient code.  Let us discuss related work along
two directions: (1) memory analysis and (2) memory optimization.
Regarding (1), the static modeling of the EVM memory in
\cite{GrechLTS22} has some similarities with the memory analysis
presented in Secs.~\ref{sec:allocation-analysis} and
\ref{sec:slots-analysis}, since in both cases we are seeking to model
the memory although with different applications in mind. There are
differences on one hand on the type of static analysis used in both
cases: \cite{GrechLTS22} is based on a Datalog analysis while we have
defined a standard transfer function which is used within a
flow-sensitive analysis. More importantly, there are differences on
the precision of both analyses. We can accurately model the memory
allocated by nested data structures in which the memory contains
pointers to other memory slots, while \cite{GrechLTS22} does not
capture such type of accesses.  This is fundamental to perform memory
optimization since, as shown in the running examples of the paper, it
allows  detecting needless write accesses that otherwise would be
missed. Finally, the application of the memory analysis to
optimization is not studied in \cite{GrechLTS22}, while it is the main
focus of our work.

As regards (2), optimizing memory usage is a challenging research
problem that requires to  
precisely infer the memory positions that
are being accessed. Such positions sometimes are statically known
(e.g., when accessing the EVM free memory pointer) but, as we have seen, often a precise
and complex inference is required to figure out the slot being
accessed at each memory access bytecode. Recent work within the
super-optimizer GASOL~\cite{AlbertGHR22} is able to perform some
memory optimizations at the level of each block of the CFG (i.e.,
intra-block). 
%
There are three fundamental differences between our work and GASOL:
First, GASOL can only apply the optimizations when the memory
locations being addressed refer to the same constant
direction. In other words, there is no real memory analysis (namely
 Secs.~\ref{sec:allocation-analysis}
and~\ref{sec:slots-analysis}). Second, the optimizations are applied
only at an intra-block level and hence many optimization opportunities are
missed.  These two points make a fundamental difference with our
approach, since 
the detected optimizable patterns (see Sec.~\ref{sec:experiments}) require inter-block analysis and a
precise slot access analysis, and hence cannot be detected by GASOL.

Finally, as mentioned in Sec.~\ref{sec:introduction}, in addition to
dynamic memory, smart contracts also use a persistent memory called
storage.  
Regarding the application of our approach to infer needless accesses
in storage, there are two main points. First, there is no need to
develop a static analysis to detect the slots in storage, as they are
statically known (hence our inference in
Sec.~\ref{sec:allocation-analysis} and \ref{sec:slots-analysis} is not
needed), i.e., one can easily know the read and write sets of
Def.~\ref{def:rw}. Thus, the read and write sets of our analysis can
be easily defined for storage. The second point is that, as storage is
persistent memory, a write storage access is not removable even if
there is no further read access within the smart contract, as it needs
to be stored for a future transaction. The removable write storage
accesses are only those that are rewritten and not read in-between the
two write accesses. Including this in our implementation is
straightforward. However, this situation is rather unusual, and we
believe that very few cases would be found and hence little
optimization can be achieved. 


\bibliographystyle{plain}
\bibliography{biblio}

\newpage

\appendix


\section{Proofs}\label{ap:proofs}

\subsection{Concrete Slot and Sound Abstraction}

We use as premise the description regarding the code generated by the
Solidity compiler at:
\url{https://docs.soliditylang.org/en/v0.8.17/internals/layout_in_memory.html}
whose main points are the following:
\begin{enumerate}
\item\label{assumption1} All memory allocations are performed using
  the same pattern: one or several instructions \code{MLOAD 0x40} that
  read the free memory pointer, followed by a sequence of instructions
  that finishes with an instruction that belongs to the set
  $\closing = \{\mcode{MSTORE 0x40, RETURN, REVERT, STOP,
    SELFDESTRUCT}\}$.

\item\label{assumption2} The free memory pointer is assigned an
  ever-increasing value in any execution of an EVM program.

\item\label{assumption3} All memory locations greater than \code{0x60}
  are accessed by computing the memory address using the free memory
  pointer at address \code{0x40} plus an optional offset greater or
  equal to zero. The offset is always smaller than the slot length.
\end{enumerate}

A concrete slot is identified by the accesses to the free memory
pointer for retrieving its base reference, i.e., the subsequence of
the program points in the execution trace that contain \code{MLOAD
  0x40} instructions reading the same value of the free memory
pointer. Concrete slots are abstracted by the set of program points of
these instructions.

The starting point of the analysis is a \emph{complete} and
\emph{sound} stack-sensitive control-flow-graph (CFG) that allow us to
represent all possible execution paths of a given EVM program. The
context-sensitivity is realized by cloning the blocks which are
reached from different contexts (state of the execution
stack). See~\cite{AlbertCGRR20bTR} for more details.

Let us now formalize the notion of concrete slot.

Given an execution trace $t\equiv t_0 \rightarrow^{*} t_n$, we use
$inst(t_k)$ to refer to the instruction of the form $pp{:}I$ executed
at $t_k$, and $stack(t_k,pos)$ to refer to the contents of the
position $pos$ in the stack after executing $t_k$. We use
$stack(t_k,top)$ to refer to the top of the stack, $stack(t_k,top-1)$
to refer to the element next to the top, and so on. Finally, we use
$freeptr(t_k)$ to refer to the value of the free memory pointer after
the execution of the instruction at $t_k$. We consider that
$freeptr(t_k)$ returns $-1$ when $inst(t_k) \in
\{\mcode{RETURN}, \mcode{REVERT},$ $\mcode{STOP},
\mcode{SELFDESTRUCT}\}$, any instruction that terminates
the execution of the EVM program.

\begin{definition}[concrete slot]
  Let us consider a program $P$ and a concrete execution trace from an
  initial state $t_0$ of the form $t \equiv t_0 \rightarrow^{*} t_n$
  with $n \geq 0$. We define $slots(t)$ as the set of distinct values
  $\{s_0,\ldots,s_k\}$ obtained from the execution of \code{MLOAD 0x40} instructions in
  $t$:
\[
  slots(t) = \{s | s \equiv stack(t_i,top) \wedge t_i\in t \wedge t_i\equiv pp{:}\mcode{MLOAD 0x40}\}.
  \]
\end{definition}

According to Assumptions~\ref{assumption1} and~\ref{assumption2},
whenever a transient slot is made permanent, the free memory pointer
is pushed forward and a different greater value is assigned to
\mem{0x40}.

\begin{definition}[slot loading instructions]
  Given a concrete execution trace $t\equiv t_0 \rightarrow^{*} t_n$
  with $n \geq 0$ and a concrete slot $s\in slots(t)$,
  we define the loading instructions of $s$ in $t$ as the multiset of
  program points defined as
  \[
  loading(s) = \{pp ~|~ t_k \in t \wedge inst(t_k) \equiv
  pp{:}\mcode{MLOAD 0x40} \wedge stack(t_k,top) \equiv s\}.
  \]
\end{definition}

\begin{definition}[abstract slot]\label{def:abstract-slot}
  Given a concrete execution trace from an initial state $t_0$ of the
  form $t \equiv t_0 \rightarrow^{*} t_n$ with $n \geq 0$ and a slot
  $s_i$, we define $\alpha(s) = \{pp ~|~ pp \in loading(s)\}$.
\end{definition}


\subsection{Proof of soundness of slot analysis}

\newcommand{\eqs}{\ensuremath{\mathcal{E}}}
\newcommand{\eq}{\ensuremath{\mathcal{X}}}

\begin{definition}[EVM analysis equations system]\label{def:equation-system}
  Given a program $P$, its stack-sensitive control-flow-graph $CFG$
  with all jumping instructions labeled with its potential
  destination blocks, a transfer function of the form $\rho(b,AS)$
  where $b$ is an EVM instruction and $AS$ an abstract state. The
  analysis equations system $\mathcal{E}(P)$ includes the following
  equations:

\[
\begin{array}{rllll}
  & pp{:}I & \text{Equations} \\ 
\hline
(1) & \mcode{JUMP/I DEST}~~~~ & \eq_{d} \sqsupseteq \eq_{pp-1} ~~~~~~~~~~~~ & \forall~d \in \mcode{DEST} \\ 
(2) & otherwise  & \eq_{pp} \sqsupseteq \rho(I,\eq_{pp-1}) ~~~~~~~~~~~~ & \\
\end{array}
\]

\end{definition}

The destination set \mcode{DEST} in Equation (1) contains one or two
destination program points depending on the instruction at $pp$.

When the constraint equation system is solved, constraint variables
over-approximate the points-to information for the program.  In particular,
variables of the form $\eq_{pp}$ store the slot information after the
execution of the instruction at $pp$.
Solving such system can be done iteratively. A na\"{\i}ve algorithm consists in
first initializing each constraint variable to $\{\emptyset\}$, and then
iteratively refining the values of these variables as follows:

\begin{enumerate}
\item substitute the current values of the constraint variables in
  the right-hand side of each constraint, and then evaluate the right-hand 
  side if needed;

\item if each constraint $\eq \sqsupseteq E$ holds, where $E$ is the value
  of the evaluation of the right-hand side of the previous step, then
  the process finishes; otherwise

\item for each $\eq \sqsupseteq E$ which does not hold, let $E'$ be the
  current value of $\eq$. Then update the current value of $\eq$ to $E
  \lub E'$. Once all these updates are (iteratively) applied we repeat
  the process at  step 1.
\end{enumerate}
Termination is guaranteed since the abstract domains used in the
analyses do not have
infinitely ascending chains. 


\begin{theorem}[soundness of slot analysis]\label{th:slots-analysis}
  Let us consider a program $P$ and the set of allocated abstract
  slots $\closed$ computed for $P$. For any trace
  $t \equiv t_0 \rightarrow^{*} t_n \in P$, we have that for all
  $s \in slots(t)$, there exists $a \in \closed$ such that $\alpha(s)
  \subseteq a$.
\end{theorem}

The slot analysis results are collected in $\closed$ at some program
points, namely the ones just before the instructions in $\closing$. We
base our proof for Theorem~\ref{th:slots-analysis} on the following
lemma that states that $\opened_{pp}$ is a sound over-approximation of
the transient slot pointed by the free memory pointer. 

\begin{lemma}
  \label{lemma2}
  Let us consider a program $P$ and $\eq_{pp}$ for each program point
  $pp$ in $P$, a solution to the equations system of
  Def.~\ref{def:equation-system} using transfer function $\nu$. For
  any trace $t \equiv t_0 \rightarrow^{*} t_n \in P$, we have that for
  each $t_k, 0 \leq k \leq n$, with $s \equiv freeptr(t_k)$ and
  $pp{:}I \equiv inst(t_k)$, there exists $a \in \eq_{pp}$ such that
  $\alpha(s) \subseteq a$.
\end{lemma}

\begin{proof}
We can reason by induction on
the value of $n$, the length of the traces $t_0 \rightarrow^* t_n$.

\noindent {\bf Base Case:} If $n=0$ then $slots(t)$ is empty and
Lemma~\ref{lemma2} trivially holds.

\medskip

\noindent {\bf Induction Hypothesis:}
We assume that Lemma~\ref{lemma2} holds for all traces of length
$m \leq n$ with $n \geq 0$: for any trace $t_0 \rightarrow^m t_m$,
with $inst(t_m)\equiv pp{:}I$ and $freeptr(t_m) \equiv s$, there
exists $a \in \eq_{pp}$ such that $\alpha(s) \subseteq a$.

\noindent {\bf Inductive Case:} Let us consider traces
$t_0 \rightarrow \cdots \rightarrow t_n \rightarrow t_{n+1}$ of length
$n+1 > 0$.
To extend the theorem to traces of length $n+1$ we reason for all
possible cases of the transfer function $\nu$ of
Def.~\ref{sec:slots-analysis-1} as follows.

\begin{enumerate}
\item[(1)] {\bf $inst(t_{n+1})\equiv pp{:}\mcode{MLOAD 0x40}$}.\\
  In this case the equation considered is
  \begin{equation}
    \label{eq:lemma2-1}
    \eq_{pp} \sqsupseteq \{s \cup \{pp\} ~|~ s \in \eq_{pp-1}\} 
  \end{equation}
  where $pp-1$ is the program point of the instruction executed
  at $t_n$. 
  
  By the induction hypothesis, given $s \equiv freeptr(t_n)$, at step
  $t_n$ there exists a set $a \in \eq_{pp-1}$ such that $\alpha(s)
  \subseteq a$.
  Since the instruction executed at $t_{n+1}$ is \code{MLOAD 0x40},
  the free memory pointer remains unchanged w.r.t. $t_n$ and,
  according to Def.~\ref{def:abstract-slot}, $\alpha(s)$ will contain
  the program points accumulated up to $t_n$, plus the program point
  $pp$. Given the Equation~\ref{eq:lemma2-1}, it is easy to see that
  $\alpha(s) \subseteq \eq_{pp}$ holds. 

\item[(2)] {\bf $inst(t_{n+1}) \in \closing$}.\\
  In this case the equation considered is
  \begin{equation}
    \label{eq:lemma2-2}
    \eq_{pp} \sqsupseteq \{\emptyset\} 
  \end{equation}
  Depending on the instruction executed at $t_{n+1}$ there are two
  sub-cases:
  \begin{itemize}
  \item $inst(t_{n+1}) \equiv \mcode{MSTORE 0x40}$. According to
    Assumption~2 detailed at the beginning of this section, the free
    memory pointer is assigned an ever-increasing value in any
    execution of an EVM program. That means that none of the \code{MLOAD
      0x40} instructions in $t_0 \rightarrow^{n+1}
    t_{n+1}$ can load the value $s\equiv freeptr(t_{n+1})$ and thus
    $\alpha(s)$ is empty and   $\alpha(s) \subseteq \eq_{pp}$
    trivially holds.
  \item $inst(t_{n+1})$ is an execution termination instruction of the
    set $\{\mcode{RETURN}, \mcode{REVERT},$
    $\mcode{STOP}, \mcode{SELFDESTRUCT}\}$. In that case,
    $s\equiv freeptr(t_{n+1})$ returns $-1$, $\alpha(s)$ is also
    empty and $\alpha(s) \subseteq \eq_{pp}$.
  \end{itemize}

\item[(3)] {\bf $inst(t_{n+1}) \not\in \closing\cup \{\mcode{MLOAD
      0x40}\}$}.\\
  Lemma~\ref{lemma2} trivially holds in this case, as the equation
  applied just propagates the abstract state from the previous
  instruction executed in any trace.
\end{enumerate}

\end{proof}

\noindent The proof for Theorem~\ref{th:slots-analysis} is sketched in the
following points. We use $pred(pp)$ to refer to the set of
program points that can precede $pp$ in an execution trace.

\begin{itemize}
\item We define the relational operator $\preceq$ on the abstract
  domain as follows: $A \preceq B$ holds iff for all $a \in A$, there
  exists $b \in B$ such that $a \subseteq b$.

\item Given a solution to the equations systems in
  Def.~\ref{def:equation-system}, it is easy to prove that $\eq_{pp'} \preceq \eq_{pp}$, with
  $pp'\in pred(pp)$, holds for every program point $pp{:}I \in P$ such
  that $I \not\equiv \mcode{MSTORE 0x40}$.

\item $\closed$, defined as
  $\bigcup_{pp \in \ppmstore} \eq_{pp-1}$, where $\ppmstore$ is
  the set of all program points $pp{:}I$ where $I{\in} \closing$,
  collects the abstract states containing the maximal sets of program
  points. 
  
\item Therefore, $\eq_{pp} \preceq \closed$ for every program point
  $pp{:}I \in P$.
\end{itemize}

\noindent It is straightforward to prove that
Theorem~\ref{th:slots-analysis} follows from the last point and
Lemma~\ref{lemma2}.

\subsection{Proof of soundness of slot access analysis}


We first define the property we aim at approximating, that is, the notion of
\emph{concrete slot access}:

\begin{definition}[concrete slot accesses]
  Given a concrete execution trace $t\equiv t_0 \rightarrow^{*} t_n$ with $n
  \geq 0$, a step of the form $inst(t_k) \equiv pp{:}I$ and
  a concrete slot $s$. We define $is\_accesed(s,pp,pos)$, which returns true if
  $s$ is accessed at $pp$ by means of the stack position $pos$. We define 
  \[
  slots\_accessed(pp,pos) = \{s ~|~ t_k \in t \wedge inst(t_k) \equiv pp{:}I \wedge is\_accessed(s,pp,pos)\}
  \]
\end{definition}

The slot access analysis starts by solving the equations system
defined at Def.~\ref{def:equation-system} such that the abstract state
AS is the \emph{memory analysis abstract state} defined at
Def.~\ref{def:memory-analysis-as} and $\rho$ is the \emph{memory
  analysis transfer function} $\tau$ defined at
Def.~\ref{def:transfer-accesses}.

The least solution of the equations system is a safe approximation of
the concrete slots accessed in a program point and we state it by
means of the following theorem.
Without loss of generality, we assume that the analysis returns a
mapping $\pi(x)$ that stores sets of program points contained in
$\closed$ instead of a set of identifiers.

\begin{theorem}[soundness of slot access analysis]
  \label{th:slot-accesses}
  Given any concrete execution trace
  $t \equiv t_0 \rightarrow^{*} t_n$, a step of the form
  $inst(t_k) \equiv pp{:}I$, a stack position $pos$, a concrete slot
  $s$ such that $s \in slots\_accessed(pp,pos)$ and the slot access
  analysis results $\pi$, we have that there exists
  $a \in \eq_{pred(pp)}(pos)$ such that $\alpha(s) \subseteq a$.
\end{theorem}

\begin{proof}
  We can reason by induction on the value of $n$, the length of the
  traces $t_0 \rightarrow^* t_n$.
  
  \noindent {\bf Base Case:} If $n=0$ then $slots\_accessed(\_,\_)$ is empty and
  the theorem trivially holds.
  
  \medskip
  
  \noindent {\bf Induction Hypothesis:}
  For any trace $t_0 \rightarrow^m t_m$ with
  $inst(t_m) \equiv pp{:}I$, given the solution of the equations
  system using $\tau$ as transfer function, we have that $\eq_{pp}$
  correctly keeps track of all potential abstract slots

  \noindent {\bf Inductive Case:} Let us consider traces $t_0
\rightarrow \cdots \rightarrow t_n \rightarrow t_{n+1}$ of length $n+1
> 0$.
To extend the theorem to traces of length $n+1$ we reason for all
possible cases of the transfer function $\tau$ of
Def.~\ref{def:transfer-accesses} as follows.

  \begin{itemize}
  \item[(1)] {\bf $inst(t_{n+1})\equiv pp{:}\mcode{MLOAD 0x40}$}.\\
    Given the soundness of Theorem~\ref{th:slots-analysis}, the
    abstract slots obtained from $\getslot(pp)$ are a sound
    over-approximation of the concrete slots.

  \item[(2)] {\bf $inst(t_{n+1})\equiv pp{:}\mcode{MLOAD}$}.\\

    In the concrete trace, a memory access might read two potential
    types of values:
    \begin{itemize}
    \item A value that corresponds to a baseref of a slot. The
      transfer function updates the top of the abstract state as
      follows:
      $\pi[t \mapsto \{m ~|~ s \in \pi(t) \wedge m {\in} \pi(s)\}]$,
      which keeps in the top of the stack $t$ all slots $m$ reachable
      from all slots available at the top of the stack before
      executing the \code{MLOAD} instruction.
      \item A generic value: the operand is simply popped from the stack.
    \end{itemize}

  \item[(3)] {\bf $inst(t_{n+1})\equiv pp{:}\mcode{MSTORE}$}.\\
    In the concrete trace, a memory access might write two potential
    types values:
    \begin{itemize}
    \item A value that corresponds to a baseref of a slot. In this
      case, the transfer function updates the abstract state as
      follows:
      $\pi [s \mapsto \pi(s)\cup \pi(t{-}1)]\backslash \{t,t{-}1\}$ $~
      \forall s {\in} \pi(t)$, which updates the entries of all slots
      available in the top of the stack with all baserefs available in
      the top-1 element of the stack
    \item A generic value: the transfer function just pop the two top
      elements of the stack.
    \end{itemize}

  \item[(4)] {\bf $inst(t_{n+1})\equiv pp{:}\mcode{SWAPi}$}.\\
    It trivially holds as it just interchange two entries of the
    abstract state.
  
  \item[(5)] {\bf $inst(t_{n+1})\equiv pp{:}\mcode{DUPi}$}.\\
    It trivially holds as it just copy two entries of the abstract
    state.

  \item[(6)] otherwise.  \\
    It trivially holds as it only removes those entries popped from
    the stack by the corresponding instruction.

\end{itemize}  

\end{proof}

Given the soundness of the two previous theorems, we directly prove
the soundness of the last corollary.

\begin{corollary}[soundness of needless write accesses inference]
  If $pp \in \writingleaks$, there is no execution of the program in
  which the slot accessed at $pp$ is read after executing the write
  instruction at $pp$.
\end{corollary}

\begin{proof}
  Straightforward given the soundness of the CFG and the soundness of
  Theorem~\ref{th:slot-accesses}.
\end{proof}

\end{document}